\documentclass[11pt]{amsart}
\usepackage{fullpage}
\usepackage{amsmath,amsfonts,amssymb,amsthm,mathscinet,enumitem,amscd,bbm,colortbl}
\usepackage{hyperref}
\usepackage[noabbrev,capitalize]{cleveref}

\allowdisplaybreaks

\definecolor{tableshade}{rgb}{0.85,0.85,0.85}

\theoremstyle{plain}
\newtheorem{theorem}{Theorem}[section]
\newtheorem{proposition}[theorem]{Proposition}
\newtheorem{lemma}[theorem]{Lemma}
\newtheorem{corollary}[theorem]{Corollary}
\theoremstyle{definition}
\newtheorem{definition}[theorem]{Definition}
\newtheorem{notation}[theorem]{Notation}
\newtheorem{procedure}[theorem]{Procedure}
\theoremstyle{remark}
\newtheorem{remark}[theorem]{Remark}
\newtheorem{example}[theorem]{Example}

\renewcommand{\epsilon}{\varepsilon}
\renewcommand{\phi}{\varphi}

\DeclareMathOperator{\abs}{abs}
\DeclareMathOperator{\ADF}{ADF}
\DeclareMathOperator{\As}{As}
\DeclareMathOperator{\Con}{Con}
\DeclareMathOperator{\gelo}{GELO}
\DeclareMathOperator{\Isom}{Isom}
\DeclareMathOperator{\mono}{mono}
\DeclareMathOperator{\Part}{Part}
\DeclareMathOperator{\Res}{Res}
\DeclareMathOperator{\Sat}{Sat}
\DeclareMathOperator{\Seq}{Seq}
\DeclareMathOperator{\Sols}{Sols}
\DeclareMathOperator{\ssac}{SSAC}
\DeclareMathOperator{\Stab}{Stab}
\DeclareMathOperator{\Wr}{Wr}

\newcommand{\mso}{[\kern-1.75pt [}
\newcommand{\msc}{]\kern-1.75pt ]}
\newcommand{\ms}[1]{\mso #1 \msc}
\newcommand{\C}{{\mathbb C}}
\newcommand{\E}{{\mathbb E}}
\newcommand{\N}{{\mathbb N}}
\newcommand{\Q}{{\mathbb Q}}
\newcommand{\Z}{{\mathbb Z}}
\newcommand{\charfunc}{{\mathbbm 1}}

\newcommand{\cP}{{\mathcal P}}
\newcommand{\cQ}{{\mathcal Q}}
\newcommand{\cW}{{\mathcal W}}
\newcommand{\cWp}{\cW^{(p)}}
\newcommand{\cWtwo}{\cW^{(2)}}
\newcommand{\cWthree}{\cW^{(3)}}
\newcommand{\fC}{{\mathfrak C}}
\newcommand{\fP}{{\mathfrak P}}
\newcommand{\card}[1]{\left|{#1}\right|}
\newcommand{\conj}[1]{\overline{#1}}
\newcommand{\expv}[1]{{\mathbf E}_{#1}}
\newcommand{\ev}{\E_f^\ell}
\newcommand{\mom}[1]{\mu_{#1,f}^\ell}
\newcommand{\smom}[1]{\tilde{\mu}_{#1,f}^\ell}
\newcommand{\lindexset}{{[2]\times[2]}}
\newcommand{\indexset}{{[p]\times\lindexset}}
\newcommand{\twoindexset}{{[2] \times\lindexset}}
\newcommand{\Eindexset}{{E\times\lindexset}}
\newcommand{\Findexset}{{F\times\lindexset}}
\newcommand{\eindexset}{{\{e\}\times\lindexset}}
\newcommand{\sums}[1]{\sum_{\substack{#1}}}
\newcommand{\floor}[1]{\left\lfloor{#1}\right\rfloor}
\newcommand{\ceil}[1]{\lceil{#1}\rceil}
\newcommand{\res}[1]{\vert_{#1}}

\title[Moments of Autocorrelation Demerit Factors]{Moments of Autocorrelation Demerit Factors of Binary Sequences}

\author{Daniel J.\ Katz}
\author{Miriam E.\ Ramirez}
\thanks{Daniel J.\ Katz is with the Department of Mathematics, California State University, Northridge.  Miriam E.\ Ramirez was with the Department of Mathematics, California State University, Northridge, USA.  This paper is based upon work of both authors supported in part by the National Science Foundation under Grants 1500856 and 1815487, and by work of Daniel J.\ Katz supported in part by the National Science Foundation under Grant 2206454.}
\date{16 August 2024}
\dedicatory{In memoriam Kai-Uwe Schmidt}

\begin{document}
\begin{abstract}
Sequences with low aperiodic autocorrelation are used in communications and remote sensing for synchronization and ranging.
The autocorrelation demerit factor of a sequence is the sum of the squared magnitudes of its autocorrelation values at every nonzero shift when we normalize the sequence to have unit Euclidean length.
The merit factor, introduced by Golay, is the reciprocal of the demerit factor.
We consider the uniform probability measure on the $2^\ell$ binary sequences of length $\ell$ and investigate the distribution of the demerit factors of these sequences.
Sarwate and Jedwab have respectively calculated the mean and variance of this distribution.
We develop new combinatorial techniques to calculate the $p$th central moment of the demerit factor for binary sequences of length $\ell$.
These techniques prove that for $p\geq 2$ and $\ell \geq 4$, all the central moments are strictly positive.
For any given $p$, one may use the technique to obtain an exact formula for the $p$th central moment of the demerit factor as a function of the length $\ell$.
Jedwab's formula for variance is confirmed by our technique with a short calculation, and we go beyond previous results by also deriving an exact formula for the skewness.
A computer-assisted application of our method also obtains exact formulas for the kurtosis, which we report here, as well as the fifth central moment.
\end{abstract}
\maketitle

\section{Introduction}
A {\it sequence} is a doubly infinite list $f=(\ldots,f_{-1},f_0,f_1,f_2,\ldots)$ of complex numbers in which only finitely many of the terms are nonzero.
We adopt this definition because we are thinking of our sequences aperiodically.
If $\ell$ is a nonnegative integer, then a {\it binary sequence of length $\ell$} is an $f=(\ldots,f_{-1},f_0,f_1,f_2,\ldots)$ in which $f_j\in\{-1,1\}$ for $j \in \{0,1,\ldots,\ell-1\}$ and $f_j=0$ otherwise.
Binary sequences are used to modulate signals in telecommunications and remote sensing \cite{Golomb,Golomb-Gong,Schroeder}.
Some applications, such as ranging, require very accurate timing.
For these applications, it is important that the sequence not resemble any time-delayed version of itself.

Our measure of resemblance is aperiodic autocorrelation.
If $f$ is a sequence and $s \in \Z$, then the {\it aperiodic autocorrelation of $f$ at shift $s$} is
\[
C_f(s) = \sum_{j \in \Z} f_{j+s} \conj{f_j}.
\]
Since $f_k=0$ for all but finitely many $k$, this sum is always defined and is nonzero for only finitely many $s$.
Note that $C_f(0)$ is the squared Euclidean norm of $f$.
For applications, we want $|C_f(s)|$ to be small compared to $C_f(0)$ for every nonzero $s\in\Z$; this distinction is what ensures proper timing.

There are two main measures for evaluating how low the autocorrelation of a sequence $f$ is at nonzero shifts.
One measure is the {\it peak sidelobe level}, which is the maximum of $|C_f(s)|$ over all nonzero $s \in \Z$; this can be regarded as an $l^\infty$ measure.
Another important measure is the {\it demerit factor}, which is an $l^2$ measure of smallness of autocorrelation.
The {\it (autocorrelation) demerit factor} of a nonzero sequence $f$ is
\begin{equation}\label{William}
\ADF(f)=\frac{\sums{s \in \Z \\ s\not=0} |C_f(s)|^2}{C_f(0)^2} = -1 +\frac{\sum_{s \in \Z} |C_f(s)|^2}{C_f(0)^2},
\end{equation}
which is the sum of the squared magnitudes of all autocorrelation values at nonzero shifts for the sequence that one obtains from $f$ by normalizing it to have unit Euclidean norm.
The {\it (autocorrelation) merit factor} is the reciprocal of the autocorrelation demerit factor; it was introduced by Golay in \cite[p.\ 450]{Golay-72} as the ``factor'' for a sequence and then as the ``merit factor'' in \cite[p.\ 460]{Golay-75}, while ``demerit factor'' appears later in the work of Sarwate \cite[p.\ 102]{Sarwate}.

Sequences with low demerit factor (equivalently, high merit factor) are highly desirable for communications and ranging applications.
For each given length $\ell$, we would like to understand the distribution of the demerit factors of binary sequences of length $\ell$, which always have $C_f(0)=\ell$, so the denominator of the last fraction in \eqref{William} is always $\ell^2$.
Thus, it is often convenient to study the numerator of the last fraction in \eqref{William}, which is the sum of the squares of all the autocorrelation values, so we define
\[
\ssac(f)=\sum_{s \in \Z} |C_f(s)|^2,
\]
so that for a binary sequence of length $\ell$ we have
\begin{equation}\label{Herman}
\ADF(f)=-1+\frac{\ssac(f)}{\ell^2}.
\end{equation}

For this entire paper, $\Seq(\ell)$ denotes the set of $2^\ell$ binary sequences of length $\ell$ with the uniform probability distribution, and the expected value of a random variable $v$ with respect to this distribution is denoted by $\ev v(f)=\expv{f \in \Seq(\ell)} (v(f))$.
The $p$th central moment of the random variable $v(f)$ as $f$ ranges over the binary sequences of length $\ell$ is denoted
\begin{equation}\label{Monique}
\mom{p} v(f)=\ev \left(v(f)-\ev v(f)\right)^p,
\end{equation}
and the $p$th standardized moment is denoted by
\[
\smom{p} v(f)=\frac{\mom{p} v(f)}{\left(\mom{2} v(f)\right)^{p/2}}.
\]
Sarwate \cite[eq.\ (13)]{Sarwate} found the mean of the demerit factor for binary sequences of a given length.
\begin{theorem}[Sarwate, 1984]\label{Sheila}
If $\ell$ is a positive integer, then $\ev \ADF(f)=1-1/\ell$.
\end{theorem}
Borwein and Lockhart \cite[pp.\ 1469--1470]{Borwein-Lockhart} showed that the variance of the demerit factor for binary sequences of length $\ell$ tends to $0$ as $\ell$ tends to infinity.
Jedwab \cite[Theorem 1]{Jedwab} gives an exact formula for the variance of the demerit factor for binary sequences of length $\ell$.
We present a formula involving a quasi-polynomial divided by the fourth power of the length that is equivalent to Jedwab's formula for the variance.
\begin{theorem}[Jedwab, 2019]\label{Jessica}
If $\ell$ is a positive integer, then 
\[
\mom{2} \ADF(f)=\begin{cases}
\frac{16\ell^3-60\ell^2+56\ell}{3\ell^4} & \text{if $\ell$ is even,} \\[4pt]
\frac{16\ell^3-60\ell^2+56\ell -12}{3\ell^4} & \text{if $\ell$ is odd.} 
\end{cases} 
\]
\end{theorem}
When one compares the calculation of the variance by Jedwab with that of the mean by Sarwate, one finds the first instance of a general principle: for each $p$, the determination of the $(p+1)$th moment is always considerably more difficult than that of the $p$th moment.
Jedwab follows the method of Aupetit et al.\ \cite{Aupetit}, which involves many multiple summations and is therefore somewhat difficult to execute precisely: Jedwab had to correct the calculation of Aupetit et al.\ to get the right formula.

In this paper, we devise a new combinatorial method for calculating the moments of the distribution of the demerit factor of binary sequences of length $\ell$.
For any given $p$, one may use the technique to obtain an exact formula for the $p$th central moment of the demerit factor as a function of the length $\ell$.
For $p=2$, this entails a short calculation that yields Jedwab's formula for variance.
To demonstrate that one can go further, we also use our formula for $p=3$ to derive an exact formula for the third central moment of $\ssac(f)$ as a quasi-polynomial function of sequence length, from which we determine (in Corollaries \ref{Sarah} and \ref{Shirley}) the third central moment and third standardized moment (skewness) of $\ADF(f)$.
\begin{theorem}\label{Theodore}
If $\ell$ is a positive integer, then
\[
\mom{3} \ADF(f)=
\begin{cases}
\frac{160\ell^4-1296\ell^3+3296\ell^2-2496\ell}{\ell^6} &  \text{if $\ell \equiv 0 \bmod 4$,} \\[4pt]
\frac{160\ell^4-1296\ell^3+3296\ell^2-2736\ell+576}{\ell^6} &  \text{if $\ell \equiv \pm 1 \bmod 4$,} \\[4pt] 
\frac{160\ell^4-1296\ell^3+3296\ell^2-2496\ell-384}{\ell^6} &  \text{if $\ell \equiv 2 \bmod 4$,}
\end{cases}
\]
and
\[
\smom{3} \ADF(f)=
\begin{cases}
\frac{6\sqrt{3}(10\ell^4-81 \ell^3+206 \ell^2-156 \ell)}{(4 \ell^3-15 \ell^2+14 \ell)^{3/2}} &  \text{if $\ell \equiv 0 \bmod 4$,} \\[4pt]
\frac{6\sqrt{3}(10\ell^4-81 \ell^3+206 \ell^2-171 \ell+36)}{(4 \ell^3-15 \ell^2+14 \ell -3)^{3/2}} &  \text{if $\ell \equiv \pm 1 \bmod 4$,} \\[4pt]
\frac{6\sqrt{3}(10\ell^4-81 \ell^3+206 \ell^2-156 \ell-24)}{(4 \ell^3-15 \ell^2+14 \ell)^{3/2}} &  \text{if $\ell \equiv 2 \bmod 4$.}
\end{cases}
\]
\end{theorem}
We also report in \cref{Fred} a computer-assisted determination of the fourth central moment of $\ssac(f)$ as a quasi-polynomial function of sequence length, from which we obtain the fourth central moment and fourth standardized moment (kurtosis) of $\ADF(f)$ (see Corollaries \ref{Geoffrey} and \ref{Curt}).
\begin{theorem}\label{Jake}
If $\ell$ is a positive integer, then $\mom{4} \ADF(f)$ is a quasi-polynomial function of $\ell$ of degree $6$ and period $120$ divided by the polynomial $\ell^8$ (see \cref{Geoffrey} for the precise function), while $\smom{4} \ADF(f)$ is a quasi-polynomial function of $\ell$ of degree $6$ and period $120$ divided by a quasi-polynomial function of $\ell$ of degree $6$ and period $2$ (see \cref{Curt} for the precise function). 
\end{theorem}
Our computer program was also able to find the fifth central moment of $\ADF$ as a quasi-polynomial function of $\ell$ of degree $7$ and period $55440$ divided by the polynomial $\ell^{10}$.
Our methods also shed light on interesting aspects of the distribution of demerit factors.
For instance, in \cref{Raphael} we show that our general theory implies that the odd central moments are always nonnegative, and we can also determine precisely when central moments are zero.
\begin{theorem}\label{Leonard}
Let $\ell$ and $p$ be positive integers.
Then $\mom{p} \ADF(f)$ is nonnegative.
Moreover, if (i) $p=1$, (ii) $p$ is odd with $p>1$ and $\ell\leq 3$, or (iii) $p$ is even and $\ell\leq 2$, then $\mom{p} \ADF(f)$ is zero; otherwise it is strictly positive.
\end{theorem}

Our method can be developed further to prove that the $p$th central moment of $\ssac$ for sequences of length $\ell$ is always a quasi-polynomial function of $\ell$ with rational coefficients.
Further developments of our method also show that in the limit as $\ell\to\infty$, all the standardized moments of the autocorrelation demerit factor tend to those of the standard normal distribution.
The additional theoretical tools used to obtain these results are introduced and explored in \cite{Katz-Ramirez-limiting}.
Golay believed that, as the sequence length tends to infinity, the merit factor of binary sequences is asymptotically bounded above by a number approximately equal to 12.32, and offered a heuristic explanation for his belief in \cite{Golay-1982}.
Golay's claim is equivalent to the claim that the demerit factor of binary sequences is asymptotically bounded below by a positive constant (approximately $1/12.32$) in the limit as sequence length tends to infinity.
The authors have been asked if knowledge of all the moments of the demerit factor could be used to provide a proof that the demerit factor is bounded away from zero.
Due to the exponential increase in the number $2^\ell$ of binary sequences of length $\ell$, the possibility of the existence of a small number of binary sequences whose demerit factor lies far below the mean is not excluded.
Indeed, it is noteworthy that there are infinite families of binary sequences whose asymptotic demerit factor tends to a value less than $1/6$ (see \cite[p.~128]{Jedwab-Katz-Schmidt}) even though the asymptotic mean value is $1$ (see \cref{Sheila}) and the asymptotic variance is $0$ (see \cref{Jessica}).
Thus, we find binary sequences that have demerit factors that lie below the mean by arbitrarily many standard deviations.

The rest of this paper is organized as follows.
\cref{Nellie} gives basic conventions, notations, and definitions, mostly concerning particular kinds of functions and partitions that are used in \cref{Monte} to obtain an exact formula for the central moments of $\ssac$ (cf.\ \cref{Sanri}).
Typically, this formula involves many similar partitions, so in \cref{Idelphonse} we devise an equivalence relation (via a group action) to organize these partitions into isomorphism classes.
This produces an easier formula for the central moments in \cref{Sanria}, but in order to apply this formula, one must find all the isomorphism classes of relevant partitions.
\cref{Scott} gives results that make this search easier and reduce the search to an algorithm.
We then prove \cref{Leonard} in \cref{Prunella}.
\cref{Veronica} is a brief exposition that applies this theory to compute the variance, thus confirming Jedwab's result in \cref{Jessica}.
\cref{Simon} follows with an exact calculation of skewness, thus proving \cref{Theodore}.
\cref{Curtis} then reports on our computer-assisted determination of the kurtosis in \cref{Jake}.
\cref{Paul} contains some technical algebraic results that are used in the proof of our formula for central moments in \cref{Monte}.
\cref{Lester} contains a result from linear algebra that supports the algorithm in \cref{Scott}.
\cref{Arthur} contains some technical combinatorial results that are used in our calculations of the variance and skewness of $\ADF$.
Appendices \ref{Isabel-proof}--\ref{Apple-proof} contain proofs of the main calculations used in \cref{Simon} to calculate the skewness of $\ADF$.

\section{Notation and definitions}\label{Nellie}

In this section, we give the basic conventions, notations, and definitions,  mostly concerning particular kinds of functions and partitions, which are used in \cref{Monte} to obtain an exact formula for the central moments (cf.\ \cref{Sanri}).

We use the convention that $\N=\{0,1,2,\ldots\}$ and $\Z_+=\{1,2,3,\ldots\}$.
If $\ell \in \N$, we write $[\ell]$ to mean $\{0, 1, \dots, \ell-1\}$.

We write multisets using blackboard bold square brackets: if $I$ is a set and $\{a_i\}_{i \in I}$ is a family of (not necessarily distinct) elements indexed by $I$, then $\ms{a_i: i \in I}$ is the multiset in which the multiplicity of $a$ in our multiset equals the number of $i \in I$ such that $a_i=a$.
In particular, if $a_1,\ldots,a_k$ is a list, then the multiplicity of $a$ in the multiset $\ms{a_1,\ldots,a_k}$ equals the number of times $a$ appears on our list.

If $S$ and $T$ are sets, then $T^S$ denotes the set of all functions from $S$ into $T$.
When we have a $k$-tuple, we sometimes abbreviate it by omitting commas and enclosing parentheses (e.g., the triple $(e,s,v)$ could be abbreviated as $(e s v)$ or even $e s v$), provided that the abbreviation does not introduce any ambiguity, e.g., when the context demands a tuple with precisely $k$ terms and each term is written with a single letter or digit.
For $p \in \N$, we shall be interested in triples of the form $(e,s,v) \in \indexset$ because terms in the calculation of the $p$th standard moment of the autocorrelation demerit factor are related to systems of $p$ equations, where each equation has two sides, and each side has two terms.  Thus, labeling each term requires a triple $(e,s,v)$, where $e \in [p]$ indicates the equation number, $s \in [2]$ indicates the side of the equation, and $v \in [2]$ indicates the position of terms on the given side (for both $s$ and $v$, we use $0$ to indicate left and $1$ to indicate right).
A certain kind of function, which we shall call an {\it assignment}, plays a critical role in telling us which values are assigned to the terms in such a system of equations.
\begin{definition}[Assignment]\label{Gennady-1}
Let $E\subseteq \N$.
An {\it assignment for $E$} is a function from $\Eindexset$ into $\N$, i.e., an element of $\N^\Eindexset$.
If $\tau$ is an assignment for $E$ and $(e,s,v) \in \Eindexset$, we typically use subscript notation $\tau_{e,s,v}$ (or just $\tau_{esv}$) rather than the parenthesized notation $\tau(e,s,v)$ to denote the evaluation of $\tau$ at $(e,s,v)$.
The following are notations for the set of all assignments for $E$ and some of its important subsets:
\begin{itemize}
\item $\As(E) = \N^\Eindexset$, the set of all assignments for $E$,
\item $\As(E,=) = \{\tau \in \As(E): \tau_{e 00} + \tau_{e 01} = \tau_{e 10} + \tau_{e 11} \text{ for every } e \in E\}$,
\item $\As(E,\ell) = \{\tau \in \As(E): \tau(\Eindexset) \subseteq [\ell]\}$, and
\item $\As(E,=,\ell) = \As(E,=) \cap \As(E,\ell)$.
\end{itemize}
\end{definition}
\begin{example}\label{Elaine}
Let $E=[2]=\{0,1\}$.  Let $\tau \colon \twoindexset \to \N$ be the function where $\tau_{000}=\tau_{001}=\tau_{100}=\tau_{101}=4$, $\tau_{010}=\tau_{110}=3$, and $\tau_{011}=\tau_{111}=5$.  Then notice that $\tau \in \As(E)=\As([2])$ because $\tau$ has $\twoindexset$ as its domain and $\N$ as its codomain.

Notice that $\tau \in \As([2],=)$ because $\tau_{000}+\tau_{001}=4+4=3+5=\tau_{010}+\tau_{011}$ and $\tau_{100}+\tau_{101}=4+4=3+5=\tau_{110}+\tau_{111}$.  Also, notice that $\tau \in \As([2],6)$ because $\tau([2])=\{3,4,5\} \subseteq [6]$.  We could also say that $\tau \in \As([2],10)$ or indeed $\tau \in \As([2],\ell)$ for every $\ell \geq 6$.  And because $\tau \in \As([2],6)$ and $\tau \in \As([2],=)$, we see that $\tau \in \As([2],=,6)$.

This assignment $\tau$ could be used to indicate one solution of the following system of equations
\begin{align*}
x_{000} + x_{001} & = x_{010} + x_{011} \\
x_{100} + x_{101} & = x_{110} + x_{111},
\end{align*}
where $\tau_{esv}$ is the value substituted for $x_{esv}$ for every $(e,s,v) \in \twoindexset$.
\end{example}
A partition of a set $A$ is a collection of nonempty, disjoint subsets of $A$ whose union is $A$.
If $\cP$ is a partition of $A$, then $\cP$ induces an equivalence relation on $A$ that is written $a_1 \equiv a_2 \pmod{\cP}$, which means that there is some class $P \in \cP$ such that $a_1,a_2 \in P$.
There is a partition naturally associated with every function, including our assignments from \cref{Gennady-1}.
\begin{definition}[Partition induced by a function]
Let $f\colon A \to B$ be a function.
The {\it partition induced by $f$} is the partition of $A$ that equals $\{f^{-1}(\{b\}): b \in B\}\smallsetminus\{\emptyset\}$, i.e., the classes of the partition are the nonempty fibers of $f$.
Equivalently, the partition induced by $f$ is the partition of $A$ where $a_1, a_2 \in A$ lie in the same class if and only if $f(a_1)=f(a_2)$.
In other words, the partition induced by $f$ is the partition that induces the equivalence relation on $A$ where $a_1$ is equivalent to $a_2$ if and only if $f(a_1)=f(a_2)$.
\end{definition}
\begin{example}\label{Fiona}
Consider the assignment $\tau \in \As([2],=,6)$ from \cref{Elaine}.  The partition induced by $\tau$ is  $\cP=\Big\{\{(0,0,0),(0,0,1),(1,0,0),(1,0,1)\}, \{(0,1,0),(1,1,0)\}, \{(0,1,1),(1,1,1)\}\Big\}$ because $\tau$ maps all the elements of $\{(0,0,0),(0,0,1),(1,0,0),(1,0,1)\}$ to the common value $4$; $\tau$ maps all elements of $\{(0,1,0),(1,1,0)\}$ to the common value $3$; and $\tau$ maps all elements of $\{(0,1,1),(1,1,1)\}$ to the common value $5$.
\end{example}
As we see in \cref{Fiona}, the partition induced by an assignment from $\As([p])$ is a partition of $\indexset$.  It will be helpful to have succinct notations for the set of all partitions of this form.
\begin{definition}\label{Odo}
If $E \subseteq \N$, then $\Part(E)=\{\cP:\cP \text{ is a partition of } \Eindexset \}$ and $\Part(p)=\Part([p])$ for $p\in\N$.
\end{definition}
We shall later find that it is important to keep track of all assignments that induce a particular partition.
\begin{definition}[Assignments inducing a partition]\label{Gennady-2}
Let $E \subseteq \N$.  If $\cP \in \Part(E)$, then
\begin{itemize}
\item $\As(\cP)=\{\tau \in \As(E): \tau_\beta = \tau_\gamma \text{ iff }  \beta \equiv \gamma \bmod{\cP} \}$, that is, the set of assignments that induce the partition $\cP$,
\item $\As(\cP,=) = \As(\cP) \cap \As(E,=)$,
\item $\As(\cP,\ell)=\As(\cP) \cap \As(E,\ell)$, and
\item $\As(\cP,=,\ell) = \As(\cP) \cap \As(E,=) \cap \As(E,\ell)$.
\end{itemize}
\end{definition}
\begin{example}\label{Gail}
Let $\tau \in \As([2],=,6)$ and $\cP$ be as defined in Examples \ref{Elaine} and \ref{Fiona}, where we saw that $\tau$ induces the partition $\cP$.  Note that $\cP \in \Part(2)$.  Furthermore, since $\tau$ induces $\cP$, we see that $\tau \in \As(\cP)$.
Because $\tau \in \As([2],=,6)=\As([2],=)\cap \As([2],6)$, we have $\tau \in \As(\cP,=)$ and $\tau \in \As(\cP,6)$ and $\tau \in \As(\cP,=,6)$.
\end{example}
Partitions that are induced by assignments in $\As(E,=)$ (where $E$ can be any finite subset of $\N$) are extremely important in our calculations, and so we give them a special name.
\begin{definition}[Satisfiable partition]\label{Sargon}
Let $E$ be a finite subset of $\N$.
A partition $\cP \in \Part(E)$ is said to be {\it satisfiable} if $\As(\cP,=)$ is nonempty.
(Equivalently, there is some $\ell \in \N$ such that $\As(\cP,=,\ell)$ is nonempty.)
We denote the set of satisfiable partitions of $\Eindexset$ as $\Sat(E)$, and for $p\in\N$ we use $\Sat(p)$ as a shorthand for $\Sat([p])$.
\end{definition}
\begin{example}\label{Gertrude}
Let $\tau$ and $\cP \in \Part(2)$ be as defined in Examples \ref{Elaine} and \ref{Fiona}, where we showed that $\tau$ induces $\cP$, so that $\tau \in \As(\cP,=)$ as noted in \cref{Gail}.  We also noted that $\cP \in \Part(2)$ in \cref{Gail}.
The nonemptiness of $\As(\cP,=)$ makes $\cP$ a satisfiable partition, so $\cP \in \Sat(2)$.
\end{example}
We have been discussing triples of the form $(e,s,v) \in \indexset$ (where $p \in \N$) because they index positions in systems of $p$ equations with two sides to each equation and two terms on each side.  Sometimes it is necessary to restrict ourselves to a subset of the equations in the system, so we define a form of restriction that applies both to the assignments of values to these equations and the partitions that the assignments induce.
\begin{definition}[Restriction of assignments, sets, and partitions]\label{Judy}
Let $F \subseteq E \subseteq \N$.
If $\tau \in \As(E)$, then we use $\tau\res{F}$ as a shorthand for $\tau\vert_{\Findexset}$ and we call $\tau\res{F}$ the {\it restriction of $\tau$ to $F$}.
For $P\subseteq \Eindexset$, the {\it restriction of $P$ to $F$}, written $P_F$, is $P\cap(\Findexset)$.
If $\cP\in \Part(E)$, then the {\it restriction of $\cP$ to $F$}, written $\cP_F$, is the partition $\{P_F : P \in \cP\} \smallsetminus \{\emptyset\}$ of $\Findexset$.
\end{definition}
\begin{example}\label{Helen}
Let $E=[2]=\{0,1\}$ and let $\cP \in \Part(E)=\Part(2)$ and $\tau \in \As(\cP,=,6)$ be as defined in Examples \ref{Elaine} and \ref{Fiona} and elaborated in Example \ref{Gail}, where it is shown that $\tau$ induces $\cP$.
Let $F=\{0\} \subseteq E$.
Then $\tau\res{F}$ is the function from $\{0\} \times \lindexset$ to $\N$ with $(\tau\res{F})_{000}=(\tau\res{F})_{001}=4$, $(\tau\res{F})_{010}=3$, and $(\tau\res{F})_{011}=5$.

Since $\cP=\{P,Q,R\}$, where $P=\{(0,0,0),(0,0,1),(1,0,0),(1,0,1)\}$, $Q=\{(0,1,0),(1,1,0)\}$, and $R=\{(0,1,1),(1,1,1)\}$, we have $\cP_F=\{P_F,Q_F,R_F\}$ where $P_F=\{(0,0,0),(0,0,1)\}$, $Q_F=\{(0,1,0)\}$, and $R_F=\{(0,1,1)\}$.
Thus, $\cP_F=\cP_{\{0\}}=\Big\{\{(0,0,0),(0,0,1)\}, \{(0,1,0)\}, \{(0,1,1)\}\Big\}$, which is a partition of $\{0\} \times \lindexset$.

Notice that $\tau\res{F}$ is in $\As(F)$ because it is a function from $F\times\lindexset$ into $\N$.
In fact, $\tau\res{F} \in \As(F,=)$ because $(\tau\res{F})_{000}+(\tau\res{F})_{001}=4+4=3+5=(\tau\res{F})_{010}+(\tau\res{F})_{011}$.
Also, $\tau\res{F} \in \As(F,6)$ because all the outputs of $\tau\res{F}$ lie in $[6]$, and so actually $\tau\res{F} \in \As(F,=,6)=\As(F,=)\cap\As(F,6)$.
Furthermore, notice that $\cP$ is the partition of $\Eindexset$ induced by $\tau$ (as shown in \cref{Gail}), and after we restrict, $\cP_F$ is the partition of $\Findexset$ induced by $\tau\res{F}$ because $\tau\res{F}$ maps all the elements of $\{(0,0,0),(0,0,1)\}$ to the common value $4$; $\tau\res{F}$ maps the element of $\{(0,1,0)\}$ to the value $3$; and $\tau\res{F}$ maps the element of $\{(0,1,1)\}$ to the value $5$.
Thus, $\tau\res{F} \in \As(\cP_F)$, and because $\tau\res{F} \in \As(F,=,6)$, we  know that $\tau\res{F} \in \As(\cP_F,=,6)=\As(\cP_F)\cap\As(F,=,6)$.
\end{example}
The observations in the last paragraph of \cref{Helen} are not coincidental.  We now show that the restriction process respects the various attributes for assignments that we specified in Definitions \ref{Gennady-1} and \ref{Gennady-2}.
\begin{lemma}\label{flower}
Let $p,\ell \in \N$, $F\subseteq E \subseteq [p]$, and $\cP \in \Part(E)$.
\begin{enumerate}[label=(\roman*)]
\item\label{flower-a} If $\tau\in\As(E)$, then $\tau\res{F} \in \As(F)$.
\item\label{flower-b} If $\tau\in\As(E,\ell)$, then $\tau\res{F} \in \As(F,\ell)$.
\item\label{flower-c} If $\tau\in\As(E,=)$, then $\tau\res{F} \in \As(F,=)$.
\item\label{flower-d} If $\tau\in\As(E,=,\ell)$, then $\tau\res{F} \in \As(F,=,\ell)$.
\item\label{flower-e} If $\tau\in\As(\cP)$, then $\tau\res{F} \in \As(\cP_F)$.
\item\label{flower-f} If $\tau\in\As(\cP,\ell)$, then $\tau\res{F} \in \As(\cP_F,\ell)$.
\item\label{flower-g} If $\tau\in\As(\cP,=)$, then $\tau\res{F} \in \As(\cP_F,=)$.
\item\label{flower-h} If $\tau\in\As(\cP,=,\ell)$, then $\tau\res{F} \in \As(\cP_F,=,\ell)$.
\end{enumerate}
\end{lemma}
\begin{proof}
First of all, \ref{flower-a} is immediate since $\tau$ is a function from $\Eindexset$ into $\N$, so $\tau\res{F}=\tau\vert_{\Findexset}$ is a function from $\Findexset$ into $\N$.
Then \ref{flower-b} follows because the image of a restriction of a function is contained in the image of the original.  
If $\tau \in \As(E,=)$, then for each $f \in F$, we have $\tau_{f00}+\tau_{f01}=\tau_{f10}+\tau_{f11}$, so $(\tau\res{F})_{f00}+(\tau\res{F})_{f01}=(\tau\res{F})_{f10}+(\tau\res{F})_{f11}$, and so $\tau\res{F} \in \As(F,=)$, thus proving \ref{flower-c}.
Then \ref{flower-d} follows from \ref{flower-b} and \ref{flower-c} because $\As(A,=,\ell)=\As(A,\ell)\cap \As(A,=)$ for $A \in \{E,F\}$.

Let $\tau\in\As(\cP)$, so that for every $\beta,\gamma \in \Findexset$, we have $\tau_\beta=\tau_\gamma$ if and only if $\beta \equiv \gamma \pmod{\cP}$, and so
$(\tau\res{F})_\beta=(\tau\res{F})_\gamma$ if and only if $\beta \equiv \gamma \pmod{\cP_F}$, so then $\tau\res{F} \in \As(\cP_F)$, thus proving \ref{flower-e}.
Then \ref{flower-f}, \ref{flower-g}, and \ref{flower-h} respectively follow from \ref{flower-e} in conjunction with \ref{flower-b}, \ref{flower-c}, \ref{flower-d}, respectively.
\end{proof}
In our combinatorial methods, the sizes and parities of sizes of classes in partitions are significant.
\begin{definition}[Partition type]
Let $\cP$ be a partition of a set $S$. Then the {\it type of $\cP$} is the multiset $\ms{|P|: P \in \cP}$.
\end{definition}
 \begin{definition}[Even partition] 
A partition is {\it even} if and only if the cardinality of each equivalence class is finite and even, i.e., if and only if every element of the partition's type is an even number.
\end{definition}
\begin{example}\label{Grace}
Let $\cP \in \Part(2)$ be the partition defined in \cref{Fiona}.  The type of $\cP$ is $\ms{4,2,2}$ since its classes have cardinalities $4$, $2$, and $2$, and since all of these cardinalities are even, this makes $\cP$ an even partition.
\end{example}
The following definition characterizes partitions of great importance in our calculations.
\begin{definition}[Globally even, locally odd (GELO) partition]\label{Hammurabi}
Let $E \subseteq \N$ and $\cP \in \Part(E)$.  We say that $\cP$ is {\it globally even, locally odd} (abbreviated {\it GELO}) to mean that $\cP$ is even but for every $e \in E$ the restricted partition $\cP_{\{e\}}$ is not even.  We denote the set of globally even, locally odd partitions of $\Eindexset $ as $\gelo(E)$.  For $p \in \N$, we write $\gelo(p)$ to mean $\gelo([p])$.
\end{definition}
\begin{example}\label{Ignatius}
Let $E=[2]=\{0,1\}$ and let $\cP$ be the partition from $\Part(E)=\Part(2)$ defined in \cref{Fiona} and elaborated in the intervening examples.
We noted that $\cP$ is even in \cref{Grace}.
In \cref{Helen}, we showed that $\cP_{\{0\}}=\Big\{\{(0,0,0),(0,0,1)\}, \{(0,1,0)\}, \{(0,1,1)\}\Big\}$ and here we also note that $\cP_{\{1\}}=\Big\{\{(1,0,0),(1,0,1)\}, \{(1,1,0)\}, \{(1,1,1)\}\Big\}$.
Both $\cP_{\{0\}}$ and $\cP_{\{1\}}$ have classes of size $1$, so they are not even.
Since $\cP$ is even, but for every $e \in E$, the restriction $\cP_{\{e\}}$ is not even, we see that $\cP$ is a globally even, locally odd partition in $\gelo(2)$.
\end{example}
When we calculate the moments of the distribution of demerit factors, it turns out that every nonzero term in our calculation corresponds to some partition combining the attributes of both Definitions \ref{Sargon} and \ref{Hammurabi}, so we name such partitions accordingly.
\begin{definition}[Contributory partition]
Let $p \in \N$.  Then $\cP \in \Part(p)$ is said to be {\it contributory} if it is globally even, locally odd and satisfiable.
We denote the set of contributory partitions of $\indexset$ as $\Con(p)$.
That is, $\Con(p)=\gelo(p)\cap\Sat(p)$.
\end{definition}
\begin{example}
Let $E=[2]=\{0,1\}$ and let $\cP$ be the partition in $\Part(E)=\Part(2)$ defined in \cref{Fiona} and elaborated in the intervening examples.
We showed that $\cP \in \Sat(2)$ in \cref{Gertrude} and that $\cP \in \gelo(2)$ in  \cref{Ignatius}, so $\cP \in \Con(2)$, i.e., $\cP$ is a contributory partition of $\twoindexset$.
\end{example}

\section{Moments from partitions}\label{Monte}

In this section, we obtain an exact formula for central moments of $\ssac(f)$, the sum of the squares of the autocorrelation values for a sequence $f$, where the moments are computed with $f$ ranging over the set $\Seq(\ell)$ of all binary sequences of a given length $\ell$ (equipped with uniform probability measure).
Recall from the Introduction that we use $\ev v(f)=\expv{f \in \Seq(\ell)} (v(f))$ to denote the expected value of a random variable $v$ depending on $f$ as $f$ ranges over $\Seq(\ell)$.
Also, recall from \eqref{Monique} that the $p$th central moment of the random variable $v(f)$ as $f$ ranges over $\Seq(\ell)$ is denoted
\[
\mom{p} v(f)=\ev \left(v(f)-\ev v(f)\right)^p.
\]
In this section, we will show the theory behind how we calculate the $p$th central moment of $\ssac$.
Since \eqref{Herman} shows that the demerit factor of a binary sequence $f$ of length $\ell$ is $\ADF(f)=-1+\ssac(f)/\ell^2$, it is easy to determine the $p$th central moment of the demerit factor from that of $\ssac$.
\cref{Sanri} provides a way of calculating central moments of the sum of squares of the autocorrelation in terms of contributory partitions and assignments.
The proof of this result requires some technical lemmas and their corollaries, which we defer to \cref{Paul}.
\begin{proposition}\label{Sanri}
For $p,\ell \in \N$, we have
\[ \mom{p} \ssac(f) = \sum_{\cP \in \Con(p)} |\As(\cP,=,\ell)|. \]
\end{proposition}
\begin{proof}
Let $f=(\ldots,f_0,f_1,f_2,\ldots)$ be a binary sequence of length $\ell$, so $f_j=0$ for $j\not\in[\ell]$.  If $E \subseteq N$ and $\tau \in \As(E)$, then we shall write $f^\tau$ to mean $\prod_{\gamma \in \Eindexset } f_{\tau_\gamma}$.
Now \cite[eq.\ (14)]{Katz16} shows that
\[
\ssac(f)=\sums{s,t,u,v \in \Z\\ s+t=u+v} f_s f_t f_u f_v = \sums{s,t,u,v \in [\ell] \\ s+t=u+v} f_s f_t f_u f_v, 
\]
which, by our notation becomes $\ssac(f)=\sum_{\tau \in \As(\{e\},=,\ell)} f^\tau$ for every $e \in \N$.  Thus,
\[
\mom{p} \ssac(f) = \ev\left(\prod_{e \in [p]} \left[\sum_{\tau^{(e)} \in \As(\{e\},=,\ell)} f^{\tau^{(e)}} - \ev\left(\sum_{\tau^{(e)} \in \As(\{e\},=,\ell)} f^{\tau^{(e)}} \right) \right] \right),
\]
so by the binomial expansion we have
\[
\mom{p} \ssac(f) = \ev\left(\sum_{E \subseteq [p]} \prod_{d \in [p]\smallsetminus E} \, \sum_{\tau^{(d)} \in \As(\{d\},=,\ell)} f^{\tau^{(d)}} (-1)^{\card{E}}  \prod_{e \in E}  \ev\left(\sum_{\tau^{(e)} \in \As(\{e\},=,\ell)} f^{\tau^{(e)}} \!\! \right)\!\! \right),
\]
and then \cref{Valerian} from \cref{Paul} gives
\begin{align*}
\mom{p} \ssac(f)
& = \ev\left(\sum_{E \subseteq [p]} \sum_{\upsilon \in \As([p]\smallsetminus E,=,\ell)} f^\upsilon (-1)^{\card{E}}  \prod_{e \in E}  \ev\left(\sum_{\tau^{(e)} \in \As(\{e\},=,\ell)} f^{\tau^{(e)}} \right) \right) \\
& = \sum_{E \subseteq [p]} (-1)^{\card{E}} \left(\sum_{\upsilon \in \As([p]\smallsetminus E,=,\ell)} \ev(f^\upsilon)\right) \left(\prod_{e \in E}  \sum_{\tau^{(e)} \in \As(\{e\},=,\ell)} \ev\left(f^{\tau^{(e)}} \right)\right),
\end{align*}
to which we apply \cref{Vladimir} of \cref{Paul} to obtain
\[
\mom{p} \ssac(f) = \sum_{E \subseteq [p]} (-1)^{\card{E}} \sum_{\tau \in \As([p],=,\ell)} \ev(f^{\tau\vert_{[p]\smallsetminus E}})  \prod_{e \in E}  \ev\left(f^{\tau\vert_{\{e\}}} \right),
\]
and then sort the tuples in $\As([p],=,\ell)$ according to which partitions they induce to obtain
\begin{align*}
\mom{p} \ssac(f)
&= \sum_{E \subseteq [p]} (-1)^{\card{E}} \sum_{\cP \in \Part(p)} \sum_{\tau \in \As(\cP,=,\ell)} \ev(f^{\tau\vert_{[p]\smallsetminus E}}) \prod_{e \in E}  \ev\left(f^{\tau\vert_{\{e\}}} \right) \\
&= \sum_{\cP \in \Part(p)} \sum_{E \subseteq [p]} (-1)^{\card{E}}  \sum_{\tau \in \As(\cP,=,\ell)} \ev(f^{\tau\vert_{[p]\smallsetminus E}})  \prod_{e \in E}  \ev\left(f^{\tau\vert_{\{e\}}} \right).
\end{align*}
Let $\cP$ be a fixed element of $\Part(p)$ and examine the value of the innermost summation for this fixed $\cP$.
If $\tau \in \As(\cP,=,\ell)$, then Lemma \ref{Kirkland} from \cref{Paul} tells us that
\[
\ev(f^{\tau\vert_{[p]\smallsetminus E}})  \prod_{e \in E}  \ev(f^{\tau\vert_{\{e\}}})
\]
equals $1$ if and only if both $\cP_{[p]\smallsetminus E}$ is even and $\cP_{\{e\}}$ is even for every $e \in E$; otherwise it is $0$.
Thus, we get a nonzero contribution if and only if both $\cP$ is even and $\cP_{\{e\}}$ is even for every $e \in E$.
So we fix an even $\cP \in \Part(p)$, and let $C$ be the set of all $c \in [p]$ such that $ \cP_{\{c\}}$ is even, so that
\[
\sum_{E \subseteq [p]} (-1) ^{\card{E}} \sum_{\tau \in \As(\cP,=,\ell)} \left[\ev(f^{\tau\vert_{[p]\smallsetminus E}}) \prod_{e \in E} \ev(f^{\tau\vert_{\{e\}}}) \right] 
= \sum_{E \subseteq C} (-1) ^{\card{E}} \card{\As(\cP,=,\ell)}.
\]
If $C$ is nonempty, then the sum on the right hand side is always zero; otherwise we get $|\As(\cP,=,\ell)|$.
Thus, we only get nonzero contributions from partitions that are globally even, locally odd, so our expectation calculation becomes
\begin{align*}
\mom{p} \ssac(f) & = \sum_{\cP \in \gelo(p)} |\As(\cP,=,\ell)| \\
& = \sum_{\cP \in \Con(p)} |\As(\cP,=\ell)|,
\end{align*}
where the final step simply eliminates $\cP$ such that $|\As(\cP,=)|=0$ (and thus $|\As(\cP,=,\ell)|=0$).
\end{proof}

\begin{remark}
\cref{Sanri} shows that each central moment is a sum of cardinalities, so it must be nonnegative.  Eventually in \cref{Lynn} we shall give the necessary and sufficient conditions for a central moment to be strictly positive.
\end{remark}

\begin{remark}
Since there is no GELO partition of $[1]\times\lindexset$, we know that $\Con(1)=\emptyset$, so then \cref{Sanri} correctly indicates that the first central moment is $0$.
\end{remark}

\section{Moments from isomorphism classes of partitions}\label{Idelphonse}

The exact formula for central moments in the previous section (cf.\ \cref{Sanri}) typically involves many similar partitions, so here we devise an equivalence relation (via a group action) to organize these partitions into classes.
This makes the moment calculations much easier, and produces a formula in \cref{Sanria}, which is used for the exact determination of variance, skewness, and kurtosis in Sections \ref{Veronica}--\ref{Curtis}.

Let $p, \ell \in \N $ be fixed and take $ \cP , \cQ \in \Part(p) $.
We wish to find the conditions that  guarantee $|\As(\cP,=,\ell)|=|\As(\cQ,=,\ell)|$.  

For example, suppose that $\cP, \cQ \in \Part(2)$ are given by
\begin{align*}
\cP & = \{ \{(000),(001), (100), (101)  \}, \{ (010), (110) \}, \{ (011), (111)  \} \}  \\
\cQ & = \{ \{(010),(011), (100), (101)  \}, \{ (000), (111) \}, \{ (001), (110)  \} \}.
\end{align*}
(Notice that $\cP$ here is the same partition $\cP$ defined in \cref{Fiona} and elaborated in the further examples of \cref{Nellie}.)
Let $\pi$ be the permutation of $\twoindexset$ with
\begin{equation}\label{Petra}
\pi(e,s,v)=\begin{cases}
(1-e,1-s,v) & \text{if $(e,s)=(0,0)$,} \\
(1-e,1-s,1-v) & \text{if $(e,s)=(0,1)$,} \\
(1-e,s,v) & \text{if $(e,s)=(1,0)$,} \\
(1-e,s,1-v) & \text{if $(e,s)=(1,1)$.}
\end{cases}
\end{equation}
We shall now show that the map $\tau \mapsto \tau\circ\pi$ is a bijection from $\As(\cP,=,\ell)$ to $\As(\cQ,=,\ell)$, which shows that $\card{\As(\cP,=,\ell)}=\card{\As(\cQ,=,\ell)}$.
If we start with $\tau \in \As(\cP,=,\ell)$, then there are distinct $A,B,C \in [\ell]$ such that
\begin{align*}
A & = \tau_{000} = \tau_{001} = \tau_{100} = \tau_{101}\\
B & = \tau_{010} = \tau_{110} \\
C & = \tau_{011} = \tau_{111},
\end{align*}
and 
\begin{align*}
A + A = \tau_{000} + \tau_{001} & = \tau_{010} + \tau_{011} = B+C\\
A + A = \tau_{100} + \tau_{101} & = \tau_{110} + \tau_{111} = B+C.
\end{align*}
Now set $\upsilon=\tau\circ\pi$; then 
\begin{align*}
(\upsilon_{000}, \upsilon_{001},\upsilon_{010},\upsilon_{011},\upsilon_{100},\upsilon_{101},\upsilon_{110},\upsilon_{111}) & =(\tau_{110},\tau_{111},\tau_{101},\tau_{100},\tau_{000},\tau_{001},\tau_{011},\tau_{010}) \\
& = (B,C,A,A,A,A,C,B),
\end{align*}
which has $\upsilon_{esv}=\upsilon_{e's'v'}$ if and only if $(e,s,v) \equiv (e',s',v') \pmod{\cQ}$ and also satisfies
\begin{align*}
B + C = \upsilon_{000} + \upsilon_{001} & = \upsilon_{010} + \upsilon_{011} = A+A\\
A + A = \upsilon_{100} + \upsilon_{101} & = \upsilon_{110} + \upsilon_{111} = C+B,
\end{align*}
so that $\tau\circ\pi \in \As(\cQ,=,\ell)$.
The effect of composing on the right with $\pi$ is to permute the indices of the assignment in a way that the system of equations for $\tau$ is transformed into a system of equations for $\upsilon$ by (i) first transposing the two equations, (ii) then transposing the sides of one of the equations, and (iii) then, for each of the two equations, transposing the two summands on the right-hand side.
Moving items in our system of equations in this manner will always yield a system that is still satisfied.
One can similarly show that if $\phi$ is an arbitrary element of $\As(\cQ,=,\ell)$, then $\phi \circ \pi^{-1} \in \As(\cP,=,\ell)$, and then one sees that the map $\tau \mapsto \tau \circ \pi$ from $\As(\cP,=,\ell)$ to $\As(\cQ,=,\ell)$ and the map $\phi \mapsto \phi \circ \pi^{-1}$ in the opposite direction are inverses of each other, and so $|\As(\cP,=,\ell)|=|\As(\cQ,=,\ell)|$.
Now notice that $\cP=\pi(\cQ)$.
This is no coincidence; we shall eventually show that $\{\tau\circ\pi : \tau \in \As(\pi(\cQ),=,\ell)\} = \As(\cQ,=,\ell)$ if $\cQ$ is any partition of $\indexset$ and $\pi$ is a permutation of $\indexset$ that respects the underlying structure of our system of equations.
This $\pi$ cannot be an arbitrary permutation of $\indexset$: for example, it is not hard to show that if $\pi$ is the transposition of $(0,0,0)$ and $(0,1,0)$, then $\As(\cQ,=,3)\not=\emptyset$ (one element is the $\upsilon$ displayed above with $A=1$, $B=0$, and $C=2$), but $\As(\pi(\cQ),=,\ell)=\emptyset$.

This leads us to the notion of isomorphic partitions.
To define this notion, we must first speak of groups of permutations of $\indexset$ and their action on $\Part(p)$ and $\As([p])$.
 
If $A$ is any set, we use $S_A$ to denote the group of all permutations of $A$, and for $p \in \N$ we use $S_p$ as a shorthand for $S_{[p]}$.
If $\pi \in S_{\indexset}$ and $P \subseteq \indexset$ and $\cQ$ is a set of subsets of $\indexset$, then we let $\pi$ act on $P$ by setting $\pi(P)=\{\pi(e,s,v): (e,s,v) \in P\}$ and we let $\pi$ act on $\cQ$ by setting $\pi(\cQ) = \{\pi(Q): Q \in \cQ\}$.  This gives an action of $\pi$ on $\Part(p)$.

If $\pi \in S_\indexset$, then the map $\pi^*$ is the permutation of $\As([p])$ given by $\tau\mapsto\tau\circ\pi$; the inverse of $\pi^*$ is $(\pi^{-1})^*$.
Thus $(\pi^*(\tau))_{e,s,v}=(\tau\circ\pi)_{e,s,v}=\tau_{\pi(e,s,v)}$ for each $(e,s,v)\in\indexset$.
For example, if $\pi \in S_{\twoindexset}$ is the permutation defined in \eqref{Petra}, then we can restrict the domain and codomain of $\pi^*$ to $\As(\cP,=,\ell)$ and $\As(\cQ,=,\ell)$, respectively, thus demonstrating $\card{\As(\cP,=,\ell)}=\card{\As(\cQ,=,\ell)}$.
However, the full group $S_\indexset$ of permutations of our indices does not respect the equations $\tau_{e00}+\tau_{e01}=\tau_{e10}+\tau_{e11}$ that are defining conditions for $\tau \in \As([p],=)$, so $\As([p],=)$ is not stable under this action.

We represent the elements of the wreath product $S_2 \Wr_{[2]} S_2$ as pairs of the form $\delta=\left((\digamma_0,\digamma_1),\sigma\right)$, where $\digamma_0,\digamma_1,\sigma \in S_2$ and $\delta$ acts on $(s,v) \in \lindexset$ by the rule
\[
\delta(s,v)=\left((\digamma_0,\digamma_1),\sigma\right)(s,v)=\left(\sigma(s),\digamma_{\sigma(s)}(v)\right) .
\]
This $S_2 \Wr_{[2]} S_2$ is isomorphic to the dihedral group of order $8$, and is a building block of the group of interest to us, which we now define.
\begin{notation}[$\cWp$]
For $p \in \N$, we use $\cWp$ to denote the wreath product $(S_2 \Wr_{[2]} S_2) \Wr_{[p]} S_p$, whose elements are pairs of the form $\pi=\left((\delta_0,\ldots,\delta_{p-1}),\epsilon\right)$, with $\delta_0,\ldots,\delta_{p-1} \in S_2 \Wr_{[2]} S_2$ and $\epsilon \in S_p$, where $\pi$ acts on the $(e,s,v) \in \indexset$ by the rule
\[
\pi(e,s,v)=\left((\delta_0,\ldots,\delta_{p-1}),\epsilon\right)(e,s,v)=\left(\epsilon(e),\delta_{\epsilon(e)}(s,v)\right).
\]
We can make this even more explicit if for each $j \in [p]$, we write
\[
\delta_j=\left((\digamma_{j,0},\digamma_{j,1}),\sigma_j\right),
\]
and then
\[
\pi=\left(\left(\left((\digamma_{0,0},\digamma_{0,1}),\sigma_0\right),\ldots,\left((\digamma_{p-1,0},\digamma_{p-1,1}),\sigma_{p-1}\right)\right),\epsilon\right),
\]
with
\[
\pi(e,s,v)=\left(\epsilon(e),\sigma_{\epsilon(e)}(s),\digamma_{\epsilon(e),\sigma_{\epsilon(e)}}(v)\right).
\]
We say that {\it $\pi$ uses the permutation $\epsilon$ to permute the equations, then the permutation $\sigma_e$ to permute the sides of equation $e$ for each $e \in [p]$, and then the permutation $\digamma_{e,s}$ to permute the places on side $s$ of equation $e$ for each $(e,s) \in [p]\times [2]$}.
\end{notation}
Since each element $\pi \in \cWp$ permutes $\indexset$ by this action, we identify each $\pi \in \cWp$ with the permutation in $S_\indexset$ that permutes $\indexset$ the same way.
Thus, $\cWp$ is regarded as a subgroup of $S_\indexset$ of order
\begin{equation}\label{Sidney}
|\cWp|=|S_p| \cdot |S_2 \Wr_{[2]} S_2|^p = p! 8^p = p! 2^{3 p}.
\end{equation}
Now we show that our subgroup $\cWp$ preserves important properties within $\As([p])$.
\begin{lemma} \label{Cactus}
Let $p, \ell \in \N$ and suppose that $\pi \in \cWp$ and $\cP \in \Part(p)$.  Then
\begin{enumerate}[label=(\roman*)]
\item\label{BalloonCactus} $\pi^*(\As([p]))=\As([p])$,
\item\label{Saguaro} $\pi^*(\As([p],\ell))=\As([p],\ell)$,
\item\label{PricklyPear} $\pi^*(\As([p],=))=\As([p],=)$,
\item\label{BarrelCactus} $\pi^*(\As([p],=,\ell))=\As([p],=,\ell)$,
\item\label{CrownCactus} $\pi^*(\As(\pi(\cP))=\As(\cP)$,
\item\label{FeatherCactus} $\pi^*(\As(\pi(\cP),\ell))=\As(\cP,\ell)$,
\item\label{MoonCactus} $\pi^*(\As(\pi(\cP),=))=\As(\cP,=)$, and
\item\label{StarCactus} $\pi^*(\As(\pi(\cP),=,\ell))=\As(\cP,=,\ell)$.
\end{enumerate}
Therefore the action of $\cWp$ on $\As([p])$ can be restricted to any of the following subsets: $\As([p],\ell)$, $\As([p],=)$, and $\As([p],=,\ell)$.
\end{lemma} 
\begin{proof}
Each part claims that there is an equality of sets, but we need only prove that the set on the left is contained in the set on the right, and then we can apply what we have proved with $\pi^{-1}$ in place of $\pi$ and $\pi(\cP)$ in place of $\cP$ to prove the other containment.  The first statement, \ref{BalloonCactus}, has already been made clear: $\pi^*$ is a permutation of $\As([p])$.

To prove \ref{Saguaro}, note that if $\tau \in \As([p],\ell)$, then $\tau(e,s,v) \in [\ell]$ for every $(e,s,v) \in \indexset$, and so of course $(\pi^*(\tau))(e,s,v)=(\tau\circ\pi)(e,s,v) \in [\ell]$ for every $(e,s,v) \in \indexset$.

To prove \ref{PricklyPear}, we define $\psi: \As([p]) \to \Z^p $ by $\psi(\tau) = (\tau_{e00}+\tau_{e01}-\tau_{e10}-\tau_{e11})_{e \in [p]}$.
Notice that if $\tau\in\As([p])$, then $\psi(\tau)=0$ if and only if $\tau\in \As([p],=)$.
Let $ \phi $ be the nontrivial homomorphism from $S_2$ onto $\Z^\times=(\{\pm 1\}, \cdot)$.
Now we define an action of $\cWp$ on $\Z^p$: if $\pi\in\cWp$ with
\[ \pi=\left(\left(\left((\digamma_{0,0},\digamma_{0,1}),\sigma_0\right),\ldots,\left((\digamma_{p-1,0},\digamma_{p-1,1}),\sigma_{p-1}\right)\right),\epsilon\right),\]
with $\epsilon \in S_p$ and all $\sigma_e$ and $\digamma_{e,s}$ in $S_2$, and if $u = (u_e)_{e \in [p] } \in \Z^p$, then $\pi \cdot u = (\phi(\sigma_{\epsilon(e)}) u_{\epsilon(e)})_{e \in [p]}$.
Notice that if $\pi \in \cWp$ and $\tau \in \As([p])$, then $\pi \cdot \psi(\tau)=\psi(\pi^*(\tau)))$.
That is, the following diagram commutes:
\[
\begin{CD}
\As([p]) @> \pi^* >> \As([p])\\
@VV \psi V @VV \psi V\\
\Z^p
@> \pi \cdot  >> \Z^p.
\end{CD}
\]
Recall that for $\tau\in\As([p])$, we have $\tau\in \As([p],=)$ if and only if $\psi(\tau)=0$.
Moreover, $\psi(\tau)=0 $ if and only if $\pi \cdot \psi(\tau) = 0$ for all $\pi \in \cWp$.
By the commutativity of our diagram $\pi \cdot \psi(\tau) = 0 $ for all $\pi \in \cWp$ if and only if $\psi(\pi^*(\tau))=0$ for all $\pi \in \cWp$, i.e., $\pi^*(\tau) \in \As([p],=)$ for all $\pi \in \cWp$.

Now \ref{BarrelCactus} follows from \ref{Saguaro} and \ref{PricklyPear} because $\As([p],=,\ell)=\As([p],\ell)\cap\As([p],=)$.

To prove \ref{CrownCactus}, suppose that $\tau \in \As(\pi(\cP))$.
Suppose $\gamma, \gamma' \in \indexset$.
Then $(\pi^*(\tau))_\gamma = (\pi^*(\tau))_{\gamma'}$ if and only if $\tau_{\pi(\gamma)} = \tau_{\pi(\gamma')}$.
Since $\tau\in \As(\pi(\cP))$, we have $\tau_{\pi(\gamma)} = \tau_{\pi(\gamma')}$ if and only if $\pi(\gamma)\equiv \pi(\gamma')  \pmod{\pi(\cP)}$.
Now, $\pi(\gamma) \equiv \pi(\gamma')  \pmod{\pi(\cP)}$ if and only if $\gamma \equiv \gamma' \pmod{\cP}$.
This shows that $\cP$ is the partition induced by $\pi^*(\tau)$, that is $\pi^*(\tau) \in \As(\cP)$.

Then \ref{FeatherCactus}, \ref{MoonCactus}, and \ref{StarCactus} respectively follow from \ref{CrownCactus} in conjunction with \ref{Saguaro}, \ref{PricklyPear}, \ref{BarrelCactus}, respectively.
\end{proof}

Recall that $\cWp$ is a subgroup of $S_\indexset$, and therefore it acts on elements of $\indexset$, and by extension on subsets of $\indexset$, that is, $\pi(P)=\{\pi(e,s,v): (e,s,v) \in P\}$ for $\pi\in\cWp$ and $P \subseteq \indexset$.
By further extension, $\cWp$ acts on sets of subsets of $\indexset$, that is, $\pi(\cQ)=\{\pi(Q): Q \in \cQ\}$ when $\pi\in\cWp$ and $\cQ$ is a set of subsets of $\indexset$.
These give actions of $\cWp$ on $\indexset$ and on $\Part(p)$.
\begin{definition}[Isomorphic partitions, isomorphism class]
Let $ p \in \N $ and $ \cP, \cQ \in \Part(p)$.  Then we say that $ \cP $ and $\cQ$ are {\it isomorphic} and write $\cP \cong \cQ$ to mean that there exists $ \pi \in \cWp $ such that $\cQ = \pi(\cP)$.  The {\it isomorphism class} of $\cP$ is the set of all partitions that are isomorphic to $\cP$.
\end{definition}
Since $\cWp$ is a group, the isomorphism relation is clearly an equivalence relation.   
We introduce the stabilizer of a partition under the action of $\cWp$; this will help us to determine the size of isomorphism classes.

\begin{notation}[$\Stab_{\cWp}(\cP)$]
Let $p \in \N$ and $ \cP \in \Part(p)$.
Then we use $ \Stab_{\cWp}(\cP)$ to denote $\{\pi \in \cWp: \pi(\cP)=\cP\}$, the stabilizer in $\cWp$  of $\cP$ under the action of $\cWp$ on $\Part(p)$.
\end{notation}
If $ p \in \N$, $\cP \in \Part(p)$, and $\fP$ is the isomorphism class of $\cP$, then the orbit of $ \cP $ under the action of $ \cWp $ is $ \fP $. Hence, by the orbit-stabilizer formula and \eqref{Sidney} we arrive at the following lemma.
\begin{lemma}\label{Stanley}
If $p \in \N$, $ \cP \in \Part(p)$, and $\fP$ is the isomorphism class of $\cP$, then 
\[|\fP| = \frac{|\cWp|}{|\Stab_{\cWp}(\cP)|}=\frac{p! 2^{3 p}}{|\Stab_{\cWp}(\cP)|}.\]
\end{lemma}

\begin{example}\label{Carl}
Later in \cref{Isidore}, we shall find that there are exactly two isomorphism classes of partitions in $\Con(2)$, with the partitions
\begin{align*}
\cP_1 & = \Big\{\{(0,0,0),(0,0,1),(1,0,0),(1,0,1)\}, \{(0,1,0),(1,1,0)\}, \{(0,1,1),(1,1,1)\}\Big\} \text{ and }\\
\cP_2 & = \Big\{\{(0,0,0),(1,0,0)\}, \{(0,0,1),(1,0,1)\}, \{(0,1,0),(1,1,0)\}, \{(0,1,1),(1,1,1)\}\Big\}
\end{align*}
being representatives of these two classes.  (Notice that $\cP_1$ here is the same as partition $\cP$ defined in \cref{Fiona} and elaborated in the further examples of Section \ref{Nellie}.)
We let $\fC_1$ (resp., $\fC_2$) be the isomorphism class containing $\cP_1$ (resp., $\cP_2$).
We use \cref{Stanley} here to show that $|\fC_1|=|\fC_2|=2^3$.

\cref{Stanley} tells us that $|\fC_j|=2^7/|\Stab(\cP_j)|$.
Thus, proving our claims about cardinalities is tantamount to proving that $|\Stab(\cP_1)|=|\Stab(\cP_2)|=2^4$.
Throughout this example, we write an element $\pi\in\cWtwo$ as
\[
\pi=\left(\left(\left((\digamma_{0,0},\digamma_{0,1}),\sigma_0\right),\ldots,\left((\digamma_{p-1,0},\digamma_{p-1,1}),\sigma_{p-1}\right)\right),\epsilon\right),
\]
where $\epsilon$, $\sigma_e$, and $\digamma_{e,s}$ are in $S_2$ for every $(e,s) \in [p]\times [2]$ and where
\[
\pi(e,s,v)=\left(\epsilon(e),\sigma_{\epsilon(e)}(s),\digamma_{\epsilon(e),\sigma_{\epsilon(e)}}(v)\right)
\]
for every $(e,s,v) \in \indexset$.

One can check that $\pi$ stabilizes $\cP_1$ if and only if both $\sigma_0$ and $\sigma_1$ are the identity permutation and $\digamma_{0,1}=\digamma_{1,1}$,  Thus, one is free to pick any values for $\epsilon$, $\digamma_{0,0}$, $\digamma_{1,0}$, and $\digamma_{0,1}$, and then the rest are determined.  So the stabilizer of $\cP_1$ in $\cWtwo$ has $2^4$ elements.

Similarly one can check that $\pi$ stabilizes $\cP_2$ if and only if $\sigma_0=\sigma_1$, $\digamma_{0,0}=\digamma_{1,0}$, and $\digamma_{0,1}=\digamma_{1,1}$.  Thus, one is free to pick any values for $\epsilon$, $\sigma_0$, $\digamma_{0,0}$, and $\digamma_{0,1}$, and then the rest are determined.  So the stabilizer of $\cP_2$ in $\cWtwo$ has $2^4$ elements.
\end{example}

In the next few results, we show that our group action applied to partitions preserves the properties of global evenness, local oddness, satisfiability, and contributoriness.
Recall that a permutation $ \pi \in \cWp $ acts on $\As([p],\ell)$ by composition on the right.
\begin{lemma}\label{Nicholas}
Let $ p \in \N $, let $ \pi=\left((\delta_0,\ldots,\delta_{p-1}),\epsilon\right) \in \cWp$, let $ \cP \in \Part(p) $, and let $e\in [p] $.
Then $\pi(\cP)_{\{\epsilon(e)\}}=\pi(\cP_{\{e\}})$, which has the same type as $\cP_{\{e\}}$.
\end{lemma}
\begin{proof}
For each $k \in [p]$, we define $ \Res_k : \Part(p) \to \Part(\{k\})$ by $\Res_k(\cP) = \cP_{\{k\}}$.
Notice that $ \pi(\eindexset) =  \{\epsilon(e) \}\times \lindexset$ because $\pi \in \cWp$.
This leads to the following commutative diagram:
\[
\begin{CD}
\Part(p) @>\pi>> \Part(p) \\
@VV\Res_e V @VV\Res_{\epsilon(e)}V\\
\Part(\{e\})
@>\pi>>  \Part(\{\epsilon(e)\}).
\end{CD}
\]
That is, $\pi(\cP)_{\{\epsilon(e)\}}=\pi(\cP_{\{e\}})$. 
Since $ \pi $ permutes the elements in $ \indexset $, we have $|P|=|\pi(P)| $ for $ P \in \cP_{\{e\}} $.
It follows that $ \pi(\cP_{\{e\}})$ is of the same type as $ \cP_{\{e\}} $.
\end{proof}
The globally even, locally odd property is stable under our action.
\begin{lemma}\label{Pineapple}
Let $p\in \N$ and $ \cP, \cQ \in \Part(p)$ with $ \cP \cong \cQ $.  Then $ \cP \in \gelo(p) $ if and only if $\cQ \in \gelo(p)$.
\end{lemma}
\begin{proof}
Since $ \cP \cong \cQ $ then $ \pi(\cP) = \cQ $ for some $ \pi \in \cWp $ where $\pi=(\delta,\epsilon)$ with $\delta \in (S_2\Wr_{[2]} S_2)^p$ and $\epsilon \in S_p$.
Notice that since $\pi$ permutes elements of the underlying set $\indexset$ of which $\cP$ and $\cQ$ are partitions, then $|P| = |\pi(P)| $ for all $ P \in \cP $.
So,  $ \cP $ is even if and only if $ \pi(\cP) $ is even.
		
Since $ \pi $ is a permutation, $\cP_{\{e\}}$ is non-even for all $ e \in [p] $ if and only if $ \pi(\cP_{\{e\}})  $ is non-even for all $ e \in [p] $.
By \cref{Nicholas} this means that $\cP_{\{e\}}  $ is non-even for all $ e \in [p] $ if and only if $ \pi(\cP)_{\{\epsilon(e)\}}  $ is non-even for all $ e \in [p] $.
Since $ \epsilon $ is a permutation of $ [p] $ this means $\cP_{\{e\}}  $ is non-even for all $ e \in [p] $ if and only if $ \pi(\cP)_{\{f\}}  $ is non-even for all $ f \in [p] $.
Thus $ \cP \in \gelo(p) $ if and only if $\pi(\cP) \in \gelo(p)$. 
\end{proof}
The property of satisfiability is also stable under our action.
\begin{lemma} \label{Celery}
Let $p,\ell \in \N$. If $\cP, \cQ \in \Part(p)$ with $\cP \cong \cQ$, then $|\As(\cP,=,\ell)|=|\As(\cQ,=,\ell)|$.
In particular, $\cP \in \Sat(p)$ if and only if $\cQ \in \Sat(p)$.
\end{lemma} 
\begin{proof}
Since $\cP \cong \cQ$, then $\cQ = \pi(\cP)$ for some $\pi \in \cWp$.
Then by \cref{Cactus}\ref{StarCactus}, we know that $\pi^*(\As(\cQ,=,\ell))=\As(\cP,=,\ell)$, and since $\pi^*$ is a permutation of $\As([p])$, we have $|(\As(\cQ,=,\ell))|=|\As(\cP,=,\ell)|$.
Since the partition $\cP$ (resp., $\cQ$) lies in $\Sat(p)$ if and only if $\As(\cP,=,\ell)$ (resp., $\As(\cQ,=,\ell)$) is nonempty for some $\ell \in \N$, we now see that $\cP$ is in $\Sat(p)$ if and only if $\cQ$ is in $\Sat(p)$.
\end{proof}
Recalling that $\Con(p)=\gelo(p)\cap\Sat(p)$, \cref{Pineapple} and \cref{Celery} show that the contributory property is also stable under our action.
\begin{corollary}\label{sticks}
Let $p, \ell \in \N $. If $ \cP, \cQ \in \Part(p)$ with $ \cP \cong \cQ $, then $ \cP \in \Con(p) $ if and only if $\cQ \in \Con(p)$, and furthermore $|\As(\cP,=,\ell)|=|\As(\cQ,=,\ell)|$.
\end{corollary}
This last result shows that each orbit in $\Part(p)$ under the action of $\cWp$ either contains only contributory partitions or no contributory partitions at all.
Since we are primarily interested in the contributory partitions and their equivalence classes, we make a name for the set of all such classes.
\begin{definition}[$\Isom(p)$]
Let $ p \in \N $. We use $\Isom(p)$ to denote the set of isomorphism classes of partitions in $ \Con(p)$.
\end{definition}
In view of \cref{sticks}, it is helpful to have a notation for the common value of $|\As(\cP,=,\ell)|$ for all partitions $\cP$ in an isomorphism class of contributory partitions.
\begin{definition}[$\Sols(\fP,\ell)$]\label{Sonia}
Let $ p, \ell \in \N $. 
If $\fP$ is any subset of $\Part(p)$ such that all partitions in $\fP$ are isomorphic to each other, we let $\Sols(\fP,\ell)$ be the common value (by \cref{sticks}) of $|\As(\cP,=,\ell)|$ for $\cP \in \fP$.
\end{definition}
We most commonly use this definition when $\fP \in \Isom(p)$.
Now our formula in \cref{Sanri} for central moments of the sum of squares of autocorrelation can be made much less unwieldy by grouping terms according to isomorphisms classes.
\begin{proposition}\label{Sanria}
If $ p,\ell \in \N$, then
\[ 
\mom{p} \ssac(f) = \sum_{ \fP \in \Isom(p)} |\fP| \Sols(\fP,\ell).
\]
\end{proposition}
\begin{proof}
By \cref{Sanri} we know
\begin{align*}
\mom{p} \ssac(f)
& = \sum_{\cP \in \Con(p)} |\As(\cP,=,\ell)| \\
& = \sum_{\fP \in \Isom(p)} \sum_{\cP \in \fP} |\As(\cP,=,\ell)|.
\end{align*}
Thus by \cref{sticks} and \cref{Sonia}
\[
\mom{p} \ssac(f) = \sum_{\fP \in \Isom(p)} |\fP| \Sols(\fP,\ell).  \qedhere
\]
\end{proof}

\section{Finding contributory partitions}\label{Scott}
In this section we give a practical method for finding a set of representatives for all the isomorphism classes of contributory partitions using a matrix formulation.
The following notion is helpful in describing and classifying contributory partitions.
\begin{definition}[Twin class, split class]
Let $ p \in \N$.  A {\it twin class} is any subset $P$ of $\indexset$ such that there exists some $(e,s) \in [p] \times [2]$ with $\{ (e,s,0), (e,s,1) \} \subseteq P$.
A {\it split class} is any subset $P$ of $\indexset$ such that there exists some $e \in [p]$ and $v,w \in [2]$ with $\{ (e,0,v), (e,1,w) \} \subseteq P$.
\end{definition}
The following result indicates some of the general features of contributory partitions that help narrow the search for them.
\begin{lemma}\label{Gideon}
Let $p\in \N$ and $\cP \in \Con(p)$.
\begin{enumerate}[label=(\roman*)]
\item\label{Epazote} For each $e\in[p]$, the restriction $\cP_{\{e\}}$ is a partition of type $\ms{2,1,1}$ or $\ms{1,1,1,1}$.  In the former case, the equivalence class with cardinality $2$ is a twin class, i.e., there is some $s \in [2]$ such that $\cP_{\{e\}} = \{\{(e,s,0),(e,s,1)\}, \{(e,1-s,0)\},\{(e,1-s,1)\}\}$.  In the latter case, $\cP_{\{e\}}= \{\{(e,0,0)\},\{(e,0,1)\},\{(e,1,0)\},\{(e,1,1)\}\}$.  No class of $\cP$ is split.
\item\label{Dill} We have $3 \leq |\cP| \leq 2 p$, and if $p$ is odd then $|\cP| \geq 4$.
\end{enumerate}
\end{lemma}
\begin{proof}
First we prove \ref{Epazote}.
Since $\cP$ is GELO, we know that for each $e \in [p]$, the restriction $\cP_{\{e\}}$ is a non-even partition of the set $\{e\}\times\lindexset$, so it cannot have only one set with all four elements.
Thus, for every $P \in \cP$, we know that $|P_{\{e\}}|\not=4$.
Suppose that there were some $e \in [p]$ and $P \in \cP$ with $|P_{\{e\}}|=3$ to show a contradiction.
Then there are some $s,v \in [2]$ such that $(e,s,0),(e,s,1),(e,1-s,v) \in P$, and there must be some $Q \in \cP$ with $Q\not=P$ and $(e,1-s,1-v) \in Q$.
Since $\cP \in \Sat(p)$, there is some $\tau \in \As(\cP,=)$ such that $\tau_{e,s,0}=\tau_{e,s,1}=\tau_{e,1-s,v}\not=\tau_{e,1-s,1-v}$ and $\tau_{e,s,0}+\tau_{e,s,1}=\tau_{e,1-s,v}+\tau_{e,1-s,1-v}$, which is a contradiction.
So for every $P \in \cP$, we have $|P_{\{e\}}|\leq 2$.

Let $e \in [p]$.
Then $\cP_{\{e\}}$ is a non-even partition (because $\cP$ is GELO) of the set $\{(e,s,v): s,v \in [2]\}$ of $4$ elements into classes of size at most $2$ (by the previous paragraph), so it must be of type $\ms{2,1,1}$ or $\ms{1,1,1,1}$.
In the former case, suppose for a contradiction that the class of size $2$ in $\cP_{\{e\}}$ contained $(e,t,v)$ and $(e,1-t,w)$ for some $t,v,w \in [2]$, so that the classes of size $1$ are $\{(e,t,1-v)\}$ and $\{(e,1-t,1-w)\}$.
Since $\cP \in \Sat(p)$, there is some $\tau \in \As(\cP,=)$ such that $\tau_{e,t,1-v}$ and $\tau_{e,1-t,1-w}$ are distinct from each other and from $\tau_{e,t,v}=\tau_{e,1-t,w}$.
But then this means that $\tau_{e,t,v}+\tau_{e,t,1-v}\not=\tau_{e,1-t,w}+\tau_{e,1-t,1-w}$, contradicting the fact that $\tau\in\As(\cP,=)$.
Thus the class of size $2$ in $\cP_{\{e\}}$ must contain two elements of the form $(e,s,0)$ and $(e,s,1)$ for some $s\in[2]$, while the other two classes are the singleton classes $\{(e,1-s,0)\}$ and $\{(e,1-s,1)\}$.
This means that no class of $\cP$ can be split.
On the other hand, if $\cP_{\{e\}}$ is of type $\ms{1,1,1,1}$, then $\cP_{\{e\}}$ is a partition of $\{(e,0,0),(e,0,1),(e,1,0),(e,1,1)\}$ into four singleton classes, which is again a non-split partition.

Now we prove \ref{Dill}.
By \ref{Epazote}, $\cP_{\{0\}}$ has at least $3$ classes, so $\cP$ must have at least three classes.
Since $\cP\in\Con(p)\subseteq\gelo(p)$, all classes in the partition $\cP$ are of even cardinality, hence they have at least $2$ elements, and the sum of all these cardinalities is $\big|\indexset\big|=4 p$, so the number $|\cP|$ of classes cannot exceed $4 p/2=2 p$.

We prove the final claim in \ref{Dill} by contraposition: we assume that $\cP$ has precisely three equivalence classes and prove that $p$ is even.
By \ref{Epazote}, for each $j \in [p]$, the set $\cP_{\{j\}}$ must be type $\ms{2,1,1}$, so every class $S$ of $\cP$ has $S_{\{j\}}\not=\emptyset$ and there must be at least one twin class $P$ in $\cP$.
So there are some $e \in [p]$ and $s \in [2]$ such that $(e,s,0), (e,s,1) \in P$.

We claim that the other two classes in $\cP$ are not twin classes.
Suppose that $Q$ is another twin class to show contradiction.
Then there would be some $f \in [p]$ and $t \in [2]$ such that $(f,t,0), (f,t,1) \in Q$ and some $u,v \in [2]$ such that $(e,1-s,u) \in Q$ and $(f,1-t,v) \in P$.
So if $R$ is the third class in $\cP$, then $(e,1-s,1-u),(f,1-t,1-v) \in R$.
We claim that there cannot be any $\tau \in \As(\cP,=)$, for if we let $A$, $B$, and $C$ be the values of $\tau_\gamma$ for $\gamma$ respectively in $P$, $Q$, and $R$, then $A$, $B$, and $C$ must be distinct (because $\tau\in\As(\cP)$) and (because $\tau\in\As(=)$), they must satisfy
\begin{align*}
2A &=B+C \\
2B &=A+C,
\end{align*}
which implies that $A=B$, contradicting the distinctness that we just noted.

So if $\cP$ has precisely three classes, then $\cP$ has one twin class and two non-twin classes.
If $S$ is a non-twin class of $\cP$, then since we have seen that $S_{\{j\}}$ is nonempty for every $j \in [p]$, we must have $|S_{\{j\}}|=1$ for every $j \in [p]$.
This implies $|S|=p$, and since a contributory partition is even, this means that $p$ is even.
\end{proof}
Now we describe a way of associating matrices to partitions of $\indexset$ that will help us find isomorphism classes of contributory partitions.
\begin{definition}[Display vector and matrix]
Let $p \in \N$ and let $P$ be a non-split class that is a subset of $\indexset$.  The {\it display vector} of $P$ is a $p\times 1$ matrix $v$ with rows indexed by $[p]$ where we write $v_e$ for the entry of row $e$.  Each entry of $v$ is drawn from the set $\{0,1_r,1_b,2,-1_r,-1_b,-2\}$, where integers with no subscript are considered uncolored while those with $r$ or $b$ subscripts are respectively considered to be colored red or blue.  We set $v_e=0$, $1_r$, $1_b$, $2$, $-1_r$, $-1_b$, or $-2$ respectively when $P\cap (\{e\}\times \lindexset)$ is $\emptyset$, $\{(e,0,0)\}$, $\{(e,0,1)\}$, $\{(e,0,0),(e,0,1)\}$, $\{(e,1,0)\}$, $\{(e,1,1)\}$, or $\{(e,1,0),(e,1,1)\}$.  A {\it display matrix} of collection of $\cP$ of non-split classes in $\indexset$ is any matrix obtained by arranging the display vectors for the classes of $\cP$ in some order and making these the columns of a $p\times |\cP|$ matrix.
\end{definition}
Note that \cref{Gideon}\ref{Epazote} shows that all classes in contributory partitions are non-split, so every contributory partition has a display matrix.
Note that the correspondence between non-split classes and their display vectors is bijective, while the correspondence between collections of such classes and their display matrices is not bijective, but becomes bijective if we consider equivalence classes of matrices modulo the action of the group of all permutations of columns.
We now look at other matrices that give less information about non-split classes and contributory partitions.
\begin{definition}[Map $\mono$, monochrome vector and matrix]
We let the {\it monochromator}, written $\mono$, be the map that takes any matrix with entries in $\{0,1_r,1_b,2,-1_r,-1_b,-2\}$ to a matrix with entries in $\{0,1,2,-1,-2\}$ by replacing instances of $1_r$ and $1_b$ with $1$ and instances of $-1_r$ and $-1_b$ with $-1$.  For $p \in \N$, the {\it monochrome vector} of a non-split class $P$ in $\indexset$ is $\mono(v)$ where $v$ is the display vector of $P$.  A {\it monochrome matrix} of a collection $\cP$ of non-split classes in $\indexset$ is $\mono(M)$ for $M$ a display matrix of $\cP$.
\end{definition}
\begin{definition}[Map $\abs$, absolute vector and matrix] We let the {\it absolutizer}, written $\abs$, be the map that takes any matrix with entries in $\{0,1,2,-1,-2\}$ to a matrix with entries in $\{0,1,2\}$ by replacing each entry with its absolute value.  For $p \in \N$, the {\it absolute vector} of a non-split class $P$ in $\indexset$ is $\abs(\mono(v))$ where $v$ is the display vector of $P$.  An {\it absolute matrix} of a collection $\cP$ of non-split classes in $\indexset$ is any matrix $\abs(\mono(M))$ for $M$ a display matrix of $\cP$.
\end{definition}
We are interested in the display, monochromatic, and absolute matrices of contributory partitions, so we now explore what conditions on these matrices make them correspond to partitions (\cref{Patrick}), GELO partitions (\cref{Gerald}), satisfiable partitions (\cref{Terrence}), and ultimately, contributory partitions (\cref{Clarence} and Corollaries \ref{Clarence-mono} and \ref{Clarence-abs}).
\begin{lemma}\label{Patrick}
Let $p\in \N$ and let $M$ be a display matrix for some collection $\cP$ of non-split classes in $\indexset$.
The set of classes that $M$ represents is a partition of $\indexset$ if and only if $M$ satisfies the following conditions:
\begin{enumerate}[label=(\roman*)]
\item Each row of $M$ has the following entries (not necessarily in the order we list them):
\begin{itemize}
\item one instance each of $2$ and $-2$, and $|\cP|-2$ instances of $0$;
\item one instance of each of $2$, $-1_r$, and $-1_b$, and $|\cP|-3$ instances of $0$;
\item one instance of each of $-2$, $1_r$, and $1_b$, and $|\cP|-3$ instances of $0$;
\item one instance of each of $1_r$, $1_b$, $-1_r$, and $-1_b$, and $|\cP|-4$ instances of $0$.
\end{itemize}
\item The sums of the absolute values of the elements (ignore color) in each column of $M$ is a positive integer.
\end{enumerate}
\end{lemma}
\begin{proof}
Let $\cP$ be the set of classes that $M$ represents.
For each $e \in [p]$, the condition on the $e$th row is tantamount to saying that each $(e,0,0)$, $(e,0,1)$, $(e,1,0)$, and $(e,1,1)$ occurs in exactly one of the classes of $\cP$.
This is true for all $e \in [p]$ if and only if $\cP$ is a disjoint covering of $\indexset$.
The condition on the columns is true if and only if each class in $\cP$ is nonempty.
\end{proof}
\begin{lemma}\label{Gerald}
Let $p\in \N$ and let $M$ be a display matrix for a collection $\cP$ consisting of non-split classes in $\indexset$.
Then $\cP$ is a GELO partition if and only if $M$ satisfies the following conditions:
\begin{enumerate}[label=(\roman*)]
\item Each row of $M$ has the following entries (not necessarily in the order we list them):
\begin{itemize}
\item one instance of each of $2$, $-1_r$, and $-1_b$, and $|\cP|-3$ instances of $0$;
\item one instance of each of $-2$, $1_r$, and $1_b$, and $|\cP|-3$ instances of $0$;
\item one instance of each of $1_r$, $1_b$, $-1_r$, and $-1_b$, and $|\cP|-4$ instances of $0$;
\end{itemize}
\item The sums of the absolute values of the elements (ignore color) in each column of $M$ is a positive even integer.
\end{enumerate}
\end{lemma}
\begin{proof}
These conditions are those of \cref{Patrick} combined with the parity condition on column sums that is equivalent to saying that every class in $\cP$ has even cardinality (global evenness) and a parity condition on row entries that makes the restriction of the partition to each equation have a class with an odd number of elements (local oddness).
\end{proof}
We now want to determine which of our matrices correspond to satisfiable partitions.  To do this, we note the connection between a monochrome matrix of a partition and counts of assignments of interest to us.
\begin{lemma}\label{Terrence}
Let $p, \ell \in \N$, let $\cP$ be a partition of $\indexset$, and let $M$ be a monochrome matrix of $\cP$.
Then $|\As(\cP,=,\ell)|$ is equal to the number of solutions $x$ in $[\ell]^{|\cP|}$ of the homogeneous linear system $M x = 0$ such that no two coordinates of $x$ have the same value.
Therefore, $\cP$ is satisfiable if and only if the homogeneous linear system $M x=0$ has a solution in $\N^{|\cP|}$ such that no two coordinates of $x$ have the same value.
\end{lemma}
\begin{proof}
Let us index the rows of $M$ by elements of $[p]$ and the columns of $M$ by classes in $\cP$, where the class $P$ indexes the column of $M$ which is the monochrome vector for $P$.
That is, $\cP$ consists of some number $t$ of classes, say $P_1,\ldots,P_t$, and so $M$ has $t$ columns that are the monochrome vectors of $P_1,\ldots,P_t$ in that order.
We identify each function $f$ from $\cP$ to $\N$ with a $|\cP|\times 1$ matrix (column vector) with entries in $\N$, where the rows of the matrix (entries of the vector) are indexed by the classes of $\cP$, arranged in the same order $P_1,\ldots,P_r$ that they are as indices of columns of $M$; in this correspondence $f(P_j)$ is the $P_j$th component of the vector, which is the $j$th component if we read out the vector from top to bottom.

Recall that an element of $\As(\cP,=,\ell)$ is a function $\tau\colon \indexset \to \N$ with $\tau(\alpha) \in [\ell]$ for every $\alpha \in \indexset$, with $\tau(\alpha)=\tau(\beta)$ if and only if $\alpha\equiv\beta \pmod{\cP}$, and such that for every $e \in [p]$ we have $\sum_{(s,v) \in \lindexset} (-1)^s \tau(e,s,v)=0$.
Let $\phi \colon \indexset \to \cP$ be the function that maps an element $\alpha$ to the equivalence class of $\cP$ that contains $\alpha$.
Let $F$ be the set of injective functions $f$ from $\cP$ to $\N$ that have $f(P)\in[\ell]$ for every $P \in \cP$ and, for every $e \in [p]$, have
\begin{equation}\label{Ricardo}
\sum_{(s,v) \in \lindexset} (-1)^s f(\phi(e,s,v))=0.
\end{equation}
Then there is a bijective correspondence between $F$ and $\As(\cP,=,\ell)$ given by $f \mapsto f \circ \phi$.
For any subset $X$ of $\indexset$, we let $\charfunc_X$ be the indicator function of $X$.
Then \eqref{Ricardo} is equivalent to $\sum_{(s,v) \in \lindexset} \sum_{P \in \cP} \charfunc_P(e,s,v) (-1)^s f(\phi(e,s,v))=0$ since every element $(e,s,v)$ of $\indexset$ occurs in one and only one class $P$ of the partition $\cP$.
This last equation, in turn, is equivalent to $\sum_{(s,v) \in \lindexset} \sum_{P \in \cP} \charfunc_P(e,s,v) (-1)^s f(P)=0$ due to the indicator function and the definition of $\phi$.
Notice that $\sum_{(s,v) \in \lindexset} (-1)^s \charfunc_P(e,s,v)$ is exactly $M_{e,P}$ because of the rules about how entries of display matrices (and therefore monochrome matrices) are determined.
Thus, \eqref{Ricardo} is equivalent to $\sum_{P \in \cP} M_{e,P} f(P)=0$.
Therefore $|\As(\cP,=,\ell)|$ is equal to the number of injective functions $f$ from $\cP$ to $[\ell]$ with the property such that for every $e \in [p]$ we have $\sum_{P \in \cP} M_{e,P} f(P)=0$.
When we identify these functions with vectors as described in the opening paragraph of this proof, this is the same as the number of $|\cP|\times 1$ column vectors $f$ with distinct entries in $[\ell]$ such that that $M f=0$.

The partition $\cP$ is satisfiable if $\As(\cP,=)$ is nonempty, which is true if and only if there is some $n\in\N$ such that $\As(\cP,=,n)$ is nonempty.  In light of what we have just proved, this is true if and only if there is some solution $x \in \N^{|\cP|}$ to the homogeneous system $M x=0$ where distinct coordinates of $x$ have different values.
\end{proof}
\cref{Terrence} gives us a concrete way of computing numbers of assignments in $|\As(\cP,=,\ell)|$ for some partition $\cP$ and $\ell \in \N$ based on a monochrome matrix for $\cP$.
We want to turn Lemmas \ref{Gerald} and \ref{Terrence} into a concrete characterization of what the display, monochrome, and absolute matrices of contributory partitions look like.
Our characterization uses the following definition.
\begin{definition}[Hamming weight and distance]\label{Hamilcar}
The {\it Hamming weight} of a matrix is the number of nonzero entries in the matrix.
The {\it Hamming distance} between two matrices of the same shape is the Hamming weight of their difference.
\end{definition}
Now we can characterize the display matrices of contributory partitions with the help of a technical lemma from \cref{Lester}.
\begin{lemma}\label{Clarence}
Let $p,t \in \N$.  A $p\times t$ matrix $M$ is a display matrix of some partition in $\Con(p)$ if and only if the following hold:
\begin{enumerate}[label=(\roman*)]
\item\label{Rowena} Each row of $M$ has one of the following as its entries (not necessarily in the order we list them):
\begin{itemize}
\item one instance of each of $2$, $-1_r$, and $-1_b$, and $t-3$ instances of $0$;
\item one instance of each of $-2$, $1_r$, and $1_b$, and $t-3$ instances of $0$; or
\item one instance of each of $1_r$, $1_b$, $-1_r$, and $-1_b$, and $t-4$ instances of $0$.
\end{itemize}
\item\label{Colin} The sum of the absolute values of the elements (ignore colors) of each column in $M$ is a positive even integer.
\item\label{Sanjay} The matrix $\mono(M)$ has a reduced row echelon form in which no row has Hamming weight $2$ and no pair of rows has Hamming distance $2$.
\end{enumerate}
Furthermore, if $M$ represents a contributory partition, then the number $t$ of columns satisfies $3 \leq t \leq 2 p$, and when $p$ is odd it also satisfies $t \geq 4$.
\end{lemma}
\begin{proof}
Let $\cP$ be the collection of non-split classes represented by $M$.
By \cref{Gerald}, conditions \ref{Rowena} and \ref{Colin} taken together are necessary and sufficient to make $\cP$ a GELO partition.
Once these are satisfied, the final condition \ref{Sanjay} is equivalent to saying that $\cP$ is satisfiable; this is due to \cref{Terrence} and \cref{Edith} (from \cref{Lester}), the latter of which applies because condition \ref{Rowena} makes every row sum of $\mono(M)$ zero.
The bounds on $t$ in the case that $\cP$ is contributory follow from \cref{Gideon}\ref{Dill}.
\end{proof}
The necessary and sufficient conditions for monochrome matrices of contributory partitions are then easy to deduce.
\begin{corollary}\label{Clarence-mono}
Let $p,t \in \N$.  A $p\times t$ matrix $M$ is a monochromatic matrix of some contributory partition if and only if the following hold:
\begin{enumerate}[label=(\roman*)]
\item\label{Rowena-mono} Each row of $M$ has one of the following as its entries (not necessarily in the order we list them):
\begin{itemize}
\item one instance of $2$, two instances of $-1$, and $t-3$ instances of $0$;
\item one instance of $-2$, two instances of $1$, and $t-3$ instances of $0$; or
\item two instances of $1$, two instances of $-1$, and $t-4$ instances of $0$.
\end{itemize}
\item The sum of the absolute values of the elements of each column in $M$ is a positive even integer.
\item\label{Sanjay-mono} The matrix $M$ has a reduced row echelon form in which no row has Hamming weight $2$ and no pair of rows has Hamming distance $2$.
\end{enumerate}
If $M$ is the monochromatic matrix of a contributory partition, then the number $t$ of columns satisfies $3 \leq t \leq 2 p$, and when $p$ is odd it also satisfies $t \geq 4$.
Furthermore, one obtains from $M$ a display matrix of a partition in $\Con(p)$ if one colors all of the $1$ and $-1$ entries of $M$ where one replaces every pair of instances of $1$ that occur within a row with one $1_r$ and one $1_b$ (in whichever order one desires) and every pair of instances of $-1$ that occur within a row with one $-1_r$ and one $-1_b$ (in whichever order one desires). 
\end{corollary}
Now we list necessary conditions for absolute matrices of contributory partitions.
\begin{corollary}\label{Clarence-abs}
Let $p,t \in \N$.  If a $p\times t$ matrix $M$ is an absolute matrix of some contributory partition $\cP \in \Con(p)$, then the following must hold:
\begin{enumerate}[label=(\roman*)]
\item The number $t$ of columns satisfies $3 \leq t \leq 2 p$, and when $p$ is odd it also satisfies $t \geq 4$.
\item Each row of $M$ has one of the following as its entries (not necessarily in the order we list them):
\begin{itemize}
\item one instance of $2$, two instances of $1$, and $t-3$ instances of $0$; or
\item four instances of $1$ and $t-4$ instances of $0$.
\end{itemize}
\item The sum of the elements of each column in $M$ is a positive even integer.
\end{enumerate}
\end{corollary}
Now we investigate how the action of our group $\cWp$ affects these matrices.  Recall that if $\pi \in \cWp$, then we can write it as 
\[
\pi=\left(\left(\left((\digamma_{0,0},\digamma_{0,1}),\sigma_0\right),\ldots,\left((\digamma_{p-1,0},\digamma_{p-1,1}),\sigma_{p-1}\right)\right),\epsilon\right),
\]
where $\epsilon \in S_p$ and $\sigma_e$, and $\digamma_{e,s}$ are in $S_2$ for every $(e,s) \in [p]\times [2]$, so that if $(e,s,v) \in \indexset$ we have
\[
\pi(e,s,v)=\left(\epsilon(e),\sigma_{\epsilon(e)}(s),\digamma_{\epsilon(e),\sigma_{\epsilon(e)}}(v)\right),
\]
and we say that $\pi$ uses the permutation $\epsilon$ to permute the equations, then the permutation $\sigma_e$ to permute the sides of equation $e$ for each $e \in [p]$, and then the permutation $\digamma_{e,s}$ to permute the places on side $s$ of equation $e$ for each $(e,s) \in [p]\times [2]$.
If we apply $\pi$ to some $P \subseteq \indexset$ with display vector $v$, the rows of $v$ are permuted by $\epsilon$; then for each $e \in [p]$, we negate row $e$ if $\sigma_e$ transposes $0$ and $1$ (otherwise we do nothing); and then for each $(e,s) \in [p]\times[2]$, if $((-1)^{s})_r$ or $((-1)^s)_b$ occurs on row $e$, then we change its color subscript to the other color if $\digamma_{e,s}$ transposes $0$ and $1$ (otherwise we do nothing).
Thus, the effect of $\pi$ on a display matrix of a contributory partition $\cP \in \Con(p)$ is to permute the rows in some way, then negate some subset (possibly empty) of the rows, then swap color labels between pairs $1_r$ and $1_b$ that occur within a row for some subset of the rows that contain such elements, and then swap color labels between pairs $-1_r$ and $-1_b$ that occur within a row for some subset of the rows that contain such elements.
Furthermore, every possible choice of row permutation, negation of rows, and color change within rows is possible.
From this, we then see that the effect of $\pi$ on a monochromatic matrix of $\cP$ is to permute rows and negate some subset of rows, where every possible choice of row permutation and negation of rows is possible.
And from this, the effect of $\pi$ on an absolute matrix of $\cP$ is to permute rows, and every row permutation is possible.
With this in mind, we let $D_p$, $M_p$, and $A_p$ respectively be the set of all display, monochromatic, and absolute matrices of partitions in $\Con(p)$, and we make an equivalence relation $\sim$ on each set that decrees that $M\sim N$ under the action of the group generated by column permutations (this just reflects different choices of what order to list the display vectors of classes), row permutations, negation of any row, swapping the color labels of any pair $1_r, 1_b$ that might occur on a row, and swapping of the color labels of any pair $-1_r,-1_b$ that might occur on a row.  Then there is a bijective map from $\Isom(p) \to D_p/\sim$ and $\mono$ induces bijection from $D_p/\sim$ to $M_p/\sim$ while $\abs$ induces a surjection from $M_p/\sim$ to $A_p/\sim$.
This suggests a procedure for finding a set of representatives for $\Isom(p)$.
\begin{procedure}\label{Natasha}
For $p \in \N$, the following procedure produces a set of partitions in $\Con(p)$ with one representative from each class in $\Isom(p)$.
\begin{enumerate}
\item Find representatives of all classes in $A_p/\sim$.
\item For each representative $R$ found in the previous step, find a set $\{S_1,\ldots,S_n\}$ of representatives for all classes in $M_p/\sim$ whose elements $\abs$ maps to $R$.  This can be done by negating two of the $1$ entries in each row of $R$ in every possible way modulo permutation of columns, permutation of rows, and negation of rows.
\item Each representative $S_j$ found in the previous step corresponds to a unique class in $D_p/\sim$ of which a representative $T$ may be obtained by any coloring of the $1$ and $-1$ entries that occur in the rows of $S_j$ so that each pair of $1$'s turns into one $1_r$ and one $1_b$ and each pair of $-1$'s turns into one $-1_r$ and one $-1_b$.
\item Each representative $T$ found in the previous step corresponds to a unique class in $\Isom(p)$, a representative of which can be found by writing down the partition whose classes have the columns of $T$ as their display vectors.
\end{enumerate}
\end{procedure}
As an example, we now use this procedure to find a set of representatives for $\Isom(2)$, which we use to compute the variance in \cref{Veronica}.
\begin{example}\label{Isidore}
We claim that there are precisely two equivalence classes, $\fC_1$ and $\fC_2$, in $\Isom(2)$, which are represented respectively by partitions
\begin{align*}
\cP_1 & = \Big\{\{(0,0,0),(0,0,1),(1,0,0),(1,0,1)\}, \{(0,1,0),(1,1,0)\}, \{(0,1,1),(1,1,1)\}\Big\} \text{ and }\\
\cP_2 & = \Big\{\{(0,0,0),(1,0,0)\}, \{(0,0,1),(1,0,1)\}, \{(0,1,0),(1,1,0)\}, \{(0,1,1),(1,1,1)\}\Big\}.
\end{align*}
To prove this claim, we use \cref{Natasha}.
We start by finding representatives modulo row and column permutations for all the absolute matrices for partitions in $\Con(2)$.
\cref{Clarence-abs} tells us that absolute matrices for partitions in $\Con(2)$ are $2\times t$ matrices with $3 \leq t \leq 4$ that have entries in $\{0,1,2\}$ where each column sum is a positive even integer and each row has either a $2$ and two $1$'s (and if there is any remaining entry, then it must equal $0$) or else four $1$'s.
When $t=3$, both rows have $2,1,1$ in some order, and the $2$'s must be stacked to avoid an odd column sum, so the only possible representative is
\[
A=\begin{pmatrix}
2 & 1 & 1 \\
2 & 1 & 1
\end{pmatrix}.
\]
When $t=4$, each row either has $2,1,1,0$ in some order or $1,1,1,1$.  One cannot mix both kinds of rows, since then there would be columns with odd sums.  The only way to obtain positive and even column sums with two rows each containing $2,1,1,0$ would be
\[
B=\begin{pmatrix}
2 & 0 & 1 & 1 \\
0 & 2 & 1 & 1
\end{pmatrix},
\]
while both rows being $1,1,1,1$ produces
\[
C=\begin{pmatrix}
1 & 1 & 1 & 1 \\
1 & 1 & 1 & 1
\end{pmatrix}.
\]
Now we find representatives of classes of monochrome matrices that map to $A$, $B$, and $C$ by $\abs$, so we are assigning signs to the elements of the matrices so as to produce rows meeting condition \ref{Rowena-mono} of \cref{Clarence-mono}, and we are looking for representatives modulo row and column permutations and negation of any selection of rows.  If $X$ is an absolute matrix, then we label the monochrome matrices that arise from assigning signs to $X$ by $X_1,\ldots,X_n$ for some $n$, and in this way we obtain
\begin{center}
\begin{tabular}{ll}
$A_1=\begin{pmatrix}
2 & -1 & -1 \\
2 & -1 & -1
\end{pmatrix}$,
&
$B_1=\begin{pmatrix}
2 & 0 & -1 & -1 \\
0 & 2 & -1 & -1
\end{pmatrix}$,
\\[20pt]
$C_1=\begin{pmatrix}
1 & 1 & -1 & -1 \\
1 & 1 & -1 & -1
\end{pmatrix}$, and
&
$C_2=\begin{pmatrix}
1 & 1 & -1 & -1 \\
1 & -1 & 1 & -1
\end{pmatrix}$.
\end{tabular}
\end{center}
These meet the first two conditions of \cref{Clarence-mono}, but then when we consider the final condition \ref{Sanjay-mono} of \cref{Clarence-mono}, we see that $B_1$ has a reduced row echelon form where a pair of rows has Hamming distance $2$ while $C_2$ has a reduced row echelon form where a row has Hamming weight $2$, while $A_1$ and $C_1$ have reduced row echelon forms that meet condition \ref{Sanjay-mono}.
To obtain representatives of classes of display matrices, we then color the $1$ and $-1$ entries of $A_1$ and $C_1$ as we wish while respecting condition \ref{Rowena} of \cref{Clarence} to obtain
\[
A_1'=\begin{pmatrix}
2 & -1_r & -1_b \\
2 & -1_r & -1_b
\end{pmatrix}
\qquad \text{and} \qquad
C_1'=\begin{pmatrix}
1_r & 1_b & -1_r & -1_b \\
1_r & 1_b & -1_r & -1_b
\end{pmatrix},
\]
which respectively represent to the partitions $\cP_1$ and $\cP_2$ that we claimed to be representatives of $\Isom(2)$ at the beginning of this example.
\end{example}
We also use \cref{Natasha} to find a set of representatives for $\Isom(3)$ in \cref{Isabel-proof}, which we use to compute the skewness in \cref{Simon}.
We also programmed a computer to find sets of representatives for $\Isom(4)$ and $\Isom(5)$ in our computer-assisted calculation of kurtosis and the fifth moment (see \cref{Curtis}).

\section{Positivity of moments}\label{Prunella}

We now use our general theory to prove that the central moments are strictly positive, except when the sequences involved are very short (in which case the central moments are zero).
\begin{theorem}\label{Lynn}
Let $\ell$ and $p$ be positive integers.
Then $\mom{p} \ADF(f)$ is nonnegative.
Moreover, if (i) $p=1$, (ii) $p$ is odd with $p>1$ and $\ell\leq 3$, or (iii) $p$ is even and $\ell\leq 2$, then $\mom{p} \ADF(f)$ is zero; otherwise it is strictly positive.  
\end{theorem}
\begin{proof}
\cref{Sanri} gives the $p$th central moment as
\begin{equation}\label{Eustace}
\mom{p} \ssac(f) = \sum_{\cP \in \Con(p)} |\As(\cP,=,\ell)|,
\end{equation}
which is clearly nonnegative.
The first central moment is trivially zero.
By \cref{Gideon}\ref{Dill}, a contributory partition $\cP$ must have at least three classes if $p$ is even and at least four classes if $p$ is odd.
If $\ell < |\cP|$, then $|\As(\cP,=,\ell)|=0$ since the tuples in distinct classes must be assigned distinct elements of $[\ell]$.
Therefore the $p$th central moment must be zero if $p$ is even and $\ell \leq 2$ or if $p$ is odd and $\ell \leq 3$.

Now suppose that $p$ is even and $\ell \geq 3$ and we shall show that the $p$th central moment is strictly positive.
Let $\cP=\{P_0,P_1,P_2\}$ be the following partition:
\begin{itemize}
\item $P_k=\{(e,0,k): e \in [p]\}$ for each $k \in [2]$ and
\item $P_2=\{(e,1,v): e \in [p], v \in [2] \}$.
\end{itemize}
Since $|P_0|=|P_1|=p$ and $|P_2|=2 p$, we see that $\cP$ is even.
Furthermore, $\cP_{\{e\}}$ is of type $\ms{2,1,1}$ for every $e \in [p]$, so $\cP$ is GELO.
Let $\tau \in \As([p])$ with $\tau_\gamma=0$, $2$, or $1$, respectively, when $\gamma \in P_0$, $P_1$, or $P_2$, respectively.
Then $\tau \in \As(\cP,=,\ell)$ because $\ell \geq 3$, so then $\cP \in \Con(p)$ and $|\As(\cP,=,\ell)| > 0$.
From \eqref{Eustace}, one sees that this makes the $p$th central moment strictly positive.

Finally, we suppose that $p>1$ is odd and $\ell \geq 4$ and show that the $p$th central moment is strictly positive.
Let $\cP=\{P_0,P_1,P_2,P_3\}$ be the following partition:
\begin{itemize}
\item $P_k = \{ (k,0,0), (k,0,1), (1-k,1,0), (e,0,k) : 2 \leq e < p \} $ for $ k\in [2] $ and 
\item $P_k = \{ (k-2, 1,1), (e,1,k-2) : 2 \leq e < p  \} $ for $k \in \{2,3\} $.
\end{itemize}
Since $|P_0|=|P_1|=p+1$ and $|P_2|=|P_3|=p-1$, the partition $\cP$ is even.
Furthermore, $\cP_{\{0\}}$ and $\cP_{\{1\}}$ are of type $\ms{2,1,1}$ and $\cP_{\{e\}}$ is of type $\ms{1,1,1,1}$ for $2 \leq e < p$, we see that $\cP$ is GELO.
Let $\tau \in \As([p])$ with $\tau_\gamma=1$, $2$, $0$, or $3$, respectively, when $\gamma \in P_0$, $P_1$, $P_2$, or $P_3$, respectively.
Then $\tau \in \As(\cP,=,\ell)$ because $\ell \geq 4$, so then $\cP \in \Con(p)$ and $|\As(\cP,=,\ell)| > 0$.
From \eqref{Eustace}, one sees that this makes the $p$th central moment strictly positive.
\end{proof}
\begin{remark}\label{Raphael}
\cref{Leonard} follows from \cref{Lynn} because when $\ell > 0$, we have $\ADF(f)=-1+\ssac(f)/\ell^2$.
\end{remark}

\section{Explicit calculation of variance}\label{Veronica}
In this section, we verify Jedwab's formula (\cref{Jessica}) for the variance of the demerit factor by finding a formula for the variance of $\ssac$, the sum of the squares of the autocorrelations.
This comes from the $p=2$ case of \cref{Sanria}, where we have explicitly determined $\Isom(2)$ in \cref{Isidore}, and for each class $\fP$ in $\Isom(2)$ we have determined the size of that class in \cref{Carl}, and where we shall determine $\Sols(\fP,\ell)$ for each isomorphism class $\fP$ in the next lemma.
Our calculation uses some basic combinatorial results that are recorded in \cref{Arthur}.
\begin{lemma}\label{Solomon}
For the classes $\fC_1$ and $\fC_2$ as defined in \cref{Isidore}, we have
\begin{align*}
\Sols(\fC_1,\ell) & =
\begin{cases}
\frac{\ell^2-2\ell}{2} & \text{if $\ell$ is even,} \\[4pt]
\frac{\ell^2-2\ell+1}{2} & \text{if $\ell$ is odd; and} \\
\end{cases}
\\[4pt]
\Sols(\fC_2,\ell) & =
\begin{cases}
\frac{2 \ell^3-9 \ell^2+10\ell}{3} & \text{if $\ell$ is even,} \\[4pt]
\frac{2 \ell^3-9 \ell^2+10\ell-3}{3} & \text{if $\ell$ is odd.}
\end{cases}
\end{align*}
\end{lemma}
\begin{proof}
For each $j \in \{1,2\}$, the partition $\cP_j$ as described in \cref{Isidore} is an element of $\fC_j$, so \cref{Sonia} tells us that $\Sols(\fC_j,\ell)=|\As(\cP_j,=,\ell)|$, which by \cref{Terrence} is the number of solutions whose coordinates have distinct values in $[\ell]$ of the homogeneous system $M x=0$ where $M$ is a monochrome matrix of $\cP_j$, examples of which can be found in the proof of \cref{Isidore} as $A_1$ and $C_1$ for $\cP_1$ and $\cP_2$, respectively.  So the rest of this proof consists of counting such solutions for the two homogeneous linear systems derived from these matrices.

For $\Sols(\fC_1,\ell)$, we need to solve the system
\begin{align*}
A+A&=B+C\\
A+A&=B+C
\end{align*}
with distinct $A, B, C \in [\ell]$.  The number of solutions is $\floor{(\ell-1)^2/2}$ by \cref{Persephone}, and is easily seen to be the quasi-polynomial that we claimed.

For $\Sols(\fC_2,\ell)$, we need to solve the system
\begin{align*}
A+B&=C+D\\
A+B&=C+D
\end{align*}
with distinct $A, B, C, D \in [\ell]$.
So we are counting the number of solutions of $A+B=C+D$ with distinct elements $A,B,C,D \in [\ell]$, so $\Sols(\fC_2,\ell)$ is given by \cref{wowzers}\ref{wowzers-two}.
\end{proof}
Since \cref{Isidore} tells us that $\Isom(2)=\{\fC_1,\fC_2\}$ and Lemmas \ref{Carl} and \ref{Solomon} tell us $|\fC_j|$ and $\Sols(\fC_j,\ell)$ for $j\in\{1,2\}$, the variance of $\ssac$ is now a straightforward calculation using \cref{Sanria} with $p=2$.
\begin{theorem}\label{Julie}
For any positive integer $\ell$,
\[
\mom{2} \ssac(f) =
\begin{cases}
\frac{16\ell^3-60\ell^2+56\ell}{3} & \text{if $\ell$ is even,} \\[4pt]
\frac{16\ell^3-60\ell^2+56\ell -12}{3} & \text{if $\ell$ is odd.} 
\end{cases}
\]
\end{theorem}
Since $\ADF(f)=-1+\ssac(f)/\ell^2$, we can divide this result by $\ell^4$ to verify \cref{Jessica} in the Introduction.

\begin{remark}
Our calculation in \cref{Julie} confirms Jedwab's calculation of the variance of the autocorrelation demerit factor \cite[Theorem 1]{Jedwab}.
The previous calculations (\cite[pp.~42-44]{Aupetit} and \cite[pp.~5--6]{Jedwab}) of the variance of the autocorrelation demerit factor of a binary sequence $f=(\ldots,f_0,f_1,f_2,\ldots)$ are expressed as sums of expectation values of products $f_{i_1} f_{i_2} \cdots f_{i_k}$ of terms from the sequence, where $i_1,\ldots,i_k \in [\ell]$ when the sequence is of length $\ell$.
The independence and identical uniform distribution of the terms in $\{1,-1\}$ makes the expectation value $1$ if the indices $i_1,\ldots,i_k$ can be arranged into pairs of identical values; otherwise the expectation value is $0$; we also use this principle in \cref{Kirkland} (which is one of the ingredients used to prove \cref{Sanri}, which in turn is used to prove \cref{Sanria}).
To make use of this principle requires one to enumerate ways in which the pairing of indices can occur.
For the earlier mean and variance calculations, it was possible to enumerate these possibilities by a more or less direct approach, but recall that the earlier attempt to compute the variance in \cite{Aupetit} had errors in the enumeration.
If one wants to compute higher moments, a direct approach would become far too complicated, so we introduced additional combinatorial and group theoretic ideas (the contributory partitions and their isomorphism classes from Sections \ref{Nellie} and \ref{Idelphonse}).

In \cite{Jedwab}, Jedwab also calculates the variance of the autocorrelation demerit factor for collections of binary sequences that satisfy various constraints (symmetric, antisymmetric, and skew-symmetric sequences).  The authors of this paper have been asked whether the methods presented here may be applied to such families of sequences.  Although we have not adapted these methods to these other ensembles, we consider it likely that a similar group-theoretic approach could be used in the calculations of their higher moments.
\end{remark}

\section{Explicit calculation of skewness}\label{Simon}

In this section, we find an explicit formula for the skewness of the demerit factor of binary sequences of length $\ell $, with independent uniformly distributed entries, as a function of $\ell$.
To do this we follow the same steps that we used when we were calculating the variance, but the calculations are more complicated for the third moment, so we state the main results here but record their proofs in Appendices \ref{Isabel-proof}--\ref{Apple-proof}.

Since the demerit factor is $\ADF(f)=-1+\ssac(f)/\ell^2$, we first find a formula for the third central moment of $ \ssac $, the sum of the squares of the autocorrelations.
This comes from the $p=3$ case of \cref{Sanria}, where we shall explicitly determine $\Isom(3)$, the set of isomorphism classes of partitions in $\Con(3)$, and for each class we shall determine the size of that class and the number of solutions attached to it.
First we need to find the partitions in $\Con(3)$.
\begin{lemma}\label{Isabel}
There are precisely eight equivalence classes, $\fC_1,\ldots,\fC_8$, in $\Isom(3)$, which are represented respectively by partitions
\begin{align*}
\cP_1 & = \Big\{ \big\{(0,0,0),(0,0,1),(1,1,0),(2,0,0)\big\}, 
                 \big\{(1,0,0),(1,0,1),(0,1,0),(2,0,1)\big\}, \\
      & \qquad   \big\{(0,1,1),(2,1,0)\big\},   
                 \big\{(1,1,1),(2,1,1)\big\} \Big\}; \\
\cP_2 & = \Big\{ \big\{(0,0,0),(0,0,1),(1,0,0),(1,0,1)\big\},   
                 \big\{(0,1,0), (2,0,0)\big\}, \\
      & \qquad   \big\{(0,1,1), (2,0,1)\big\},   
                 \big\{(1,1,0), (2,1,0)\big\},   
                 \big\{(1,1,1), (2,1,1)\big\} \Big\}; \\
\cP_3 & = \Big\{ \big\{(0,0,0), (0,0,1), (1,1,0), (2,1,0)\big\},   
                 \big\{(1,0,0), (1,0,1)\big\}, \\
      & \qquad   \big\{(2,0,0), (2,0,1)\big\},   
                 \big\{(0,1,0), (1,1,1)\big\},   
                 \big\{(0,1,1), (2,1,1)\big\} \Big\}; \\
\cP_4 & = \Big\{ \big\{(0,0,0), (0,0,1), (1,0,0), (2,0,0)\big\}, 
                 \big\{(0,1,0), (1,1,0)\big\}, \\
      & \qquad   \big\{(0,1,1), (2,1,0)\big\},   
                 \big\{(1,0,1), (2,1,1)\big\},   
                 \big\{(1,1,1), (2,0,1)\big\} \Big\}; \\
\cP_5 & = \Big\{ \big\{(0,0,0), (0,0,1)\big\},  
                 \big\{(1,0,0), (1,0,1)\big\},   
                 \big\{(2,0,0), (2,0,1)\big\}, \\
      & \qquad   \big\{(1,1,0), (2,1,1)\big\},
                 \big\{(2,1,0), (0,1,1)\big\},
                 \big\{(0,1,0), (1,1,1)\big\} \Big\}; \\
\cP_6 & = \Big\{ \big\{(0,0,0), (0,0,1)\big\},  
                 \big\{(1,0,0), (1,0,1)\big\},
                 \big\{(0,1,0), (2,0,0)\big\}, \\ 
      & \qquad   \big\{(0,1,1), (2,1,0)\big\}, 
                 \big\{(1,1,0), (2,0,1)\big\},
                 \big\{(1,1,1), (2,1,1)\big\} \Big\}; \\
\cP_7 & = \Big\{ \big\{(0,0,0), (1,1,0)\big\},  
                 \big\{(0,0,1), (1,1,1)\big\},   
                 \big\{(1,0,0), (2,1,0)\big\}, \\
      & \qquad   \big\{(1,0,1), (2,1,1)\big\},   
                 \big\{(2,0,0), (0,1,0)\big\},   
                 \big\{(2,0,1), (0,1,1)\big\} \Big\}\text{; and} \\
\cP_8 & = \Big\{ \big\{(0,0,0), (1,1,1)\big\},  
                 \big\{(0,1,0), (1,0,1)\big\},   
                 \big\{(1,0,0), (2,1,1)\big\}, \\
     & \qquad    \big\{(1,1,0), (2,0,1)\big\},   
                 \big\{(2,0,0), (0,1,1)\big\},   
                 \big\{(2,1,0), (0,0,1)\big\} \Big\}.
\end{align*}
\end{lemma}
The proof is given in \cref{Isabel-proof}.  Now we find the cardinalities of the classes in $\Isom(3)$.
\begin{lemma}\label{Genevieve}
Let $\fC_1,\ldots,\fC_8$ be the classes of $\Isom(3)$ as described in \cref{Isabel}.
We have $\card{\fC_1}=3 \cdot 2^7$, $\card{\fC_2}=3 \cdot 2^5$, $\card{\fC_3}=3 \cdot 2^6$, $\card{\fC_4}=3 \cdot 2^8$, $\card{\fC_5}=2^6$, $\card{\fC_6}=3\cdot 2^6$, $\card{\fC_7}=2^6$, and $\card{\fC_8}=2^8$.
\end{lemma}
The proof is given in \cref{Genevieve-proof}.
To use \cref{Sanria}, we need to determine $\Sols(\fC)$ for each class $\fC$ in $\Isom(3)$.
\begin{lemma} \label{Apple}
For $\fC_1,\ldots,\fC_8$ as defined in \cref{Isabel}, we have
\begin{align*}
\Sols(\fC_1,\ell) & =
\begin{cases}
\frac{\ell^2-3\ell}{3} &  \text{if $\ell \equiv 0 \pmod{3}$,} \\[4pt]
\frac{\ell^3-3\ell+2}{3} & \text{if $\ell \equiv \pm 1 \pmod{3}$;}
\end{cases}
\\[4pt]
\Sols(\fC_2,\ell) & =
\begin{cases}
\frac{\ell^3-6\ell^2+8\ell}{3} &  \text{if $\ell \equiv 0 \pmod{2}$,} \\[4pt]
\frac{\ell^3-6\ell^2+11\ell -6}{3} & \text{if $\ell \equiv 1 \pmod{2}$;}
\end{cases}
\\[4pt]
\Sols(\fC_3,\ell) & =
\begin{cases}
\frac{\ell^2-4\ell}{4} & \text{if $\ell \equiv 0 \pmod{4}$,} \\[4pt]
\frac{\ell^2-4\ell+3}{4} & \text{if $\ell \equiv \pm 1 \pmod{4}$,}\\[4pt]
\frac{\ell^2-4\ell+4}{4} & \text{if $\ell \equiv 2 \pmod{4}$;} 
\end{cases}
\\[4pt]
\Sols(\fC_4,\ell) & =
\begin{cases}
\frac{5\ell^3-32\ell^2+52\ell}{12} & \text{if $\ell \equiv 0 \pmod{6}$,} \\[4pt]
\frac{5\ell^3-32\ell^2+55\ell-28}{12} & \text{if $\ell \equiv \pm 1 \pmod{6}$,} \\[4pt]
\frac{5\ell^3-32\ell^2+52\ell-16}{12} & \text{if $\ell \equiv \pm 2 \pmod{6}$,} \\[4pt]
\frac{5\ell^3-32\ell^2+55\ell-12}{12} & \text{if $\ell \equiv 3 \pmod{6}$;}
\end{cases}
\\[4pt]
\Sols(\fC_5,\ell) & =
\begin{cases}
\frac{\ell^3-9\ell^2+20\ell}{4} & \text{if $\ell \equiv 0 \pmod{4}$,} \\[4pt]
\frac{\ell^3-9\ell^2+23\ell-15}{4} & \text{if $\ell \equiv \pm 1 \pmod{4}$,} \\[4pt]
\frac{\ell^3-9\ell^2+20\ell-12}{4} & \text{if $\ell \equiv 2 \pmod{4}$;}
\end{cases}
\\[4pt]
\Sols(\fC_6,\ell) & =
\begin{cases}
\frac{\ell^3-8\ell^2+17\ell}{3} & \text{if $\ell \equiv 0 \pmod{12}$,} \\[4pt]
\frac{\ell^3-8\ell^2+17\ell-10}{3} & \text{if $\ell \equiv \pm 1, \pm 2, \pm 5 \pmod{12}$,} \\[4pt]
\frac{\ell^3-8\ell^2+17\ell-6}{3}  & \text{if $\ell \equiv \pm 3, 6 \pmod{12}$,} \\[4pt]
\frac{\ell^3-8\ell^2+17\ell-4}{3} & \text{if $\ell \equiv \pm 4 \pmod{12}$;}
\end{cases}
\\[4pt]
\Sols(\fC_7,\ell) & =
\begin{cases}
\frac{\ell^4-10\ell^3+32\ell^2-32\ell}{2} & \text{if $\ell \equiv 0 \pmod{2}$,} \\[4pt]
\frac{\ell^4-10\ell^3+32\ell^2-38\ell+15}{2} & \text{if $\ell \equiv 1 \pmod{2}$; and}
\end{cases} 
\\[4pt]
\Sols(\fC_8,\ell) & =
\begin{cases}
\frac{\ell^4-11\ell^3+39\ell^2-46\ell}{2} & \text{if $\ell \equiv 0 \pmod{6}$,} \\[4pt]
\frac{\ell^4-11\ell^3+39\ell^2-49\ell+20}{2} & \text{if $\ell \equiv \pm 1 \pmod{6}$,} \\[4pt]
\frac{\ell^4-11\ell^3+39\ell^2-46\ell+8}{2} & \text{if $\ell \equiv \pm 2 \pmod{6}$,} \\[4pt]
\frac{\ell^4-11\ell^3+39\ell^2-49\ell+12}{2} & \text{if $\ell \equiv 3 \pmod{6}$.}
\end{cases}
\end{align*}
\end{lemma}
The proof is given in \cref{Apple-proof}.
Since \cref{Isabel} tells us that $\Isom(3)=\{\fC_1,\ldots,\fC_8\}$ and Lemmas \ref{Genevieve} and \ref{Apple} tell us $|\fC_j|$ and $\Sols(\fC_j,\ell)$ for $j\in\{1,\ldots,8\}$, the third central moment of $\ssac$ is now a straightforward calculation using \cref{Sanria} with $p=3$.
\begin{theorem}
If $\ell \in \N$  then 
\[
\mom{3}\ssac(f) = \begin{cases}
160\ell^4-1296\ell^3+3296\ell^2-2496\ell &  \text{if $\ell \equiv 0 \pmod{4}$,} \\
160\ell^4-1296\ell^3+3296\ell^2-2736\ell+576 & \text{if $\ell \equiv \pm 1 \pmod{4}$,} \\ 
160\ell^4-1296\ell^3+3296\ell^2-2496\ell-384 & \text{if $\ell \equiv 2 \pmod{4}$.}
\end{cases} 
\]
\end{theorem}
Since $\ADF(f)=-1+\ssac(f)/\ell^2$, we can divide this result by $\ell^6$ to obtain the third central moment of the demerit factor.
\begin{corollary}\label{Sarah}
If $\ell \in \Z_+$, then 
\[
\mom{3} \ADF(f) =
\begin{cases}
\frac{160\ell^4-1296\ell^3+3296\ell^2-2496\ell}{\ell^6} & \text{if $\ell \equiv 0 \pmod{4}$,} \\[4pt]
\frac{160\ell^4-1296\ell^3+3296\ell^2-2736\ell+576}{\ell^6} & \text{if $\ell \equiv \pm 1 \pmod{4}$,} \\[4pt] 
\frac{160\ell^4-1296\ell^3+3296\ell^2-2496\ell-384}{\ell^6} & \text{if $\ell \equiv 2 \pmod{4}$.}
\end{cases}
\]
\end{corollary}
We can normalize the third central moment of $\ssac(f)$ using our variance calculation from \cref{Julie} to obtain the skewness of $\ssac(f)$, which is the same as the skewness of $\ADF(f)=-1+\ssac(f)/\ell^2$.
\begin{corollary}\label{Shirley}
If $\ell \in \Z_+$, then
\begin{equation*}
\smom{3} \ADF(f)=\smom{3}\ssac(f)=
\begin{cases}
\frac{6\sqrt{3}(10\ell^4-81 \ell^3+206 \ell^2-156 \ell)}{(4 \ell^3-15 \ell^2+14 \ell)^{3/2}} & \text{if $\ell \equiv 0 \pmod{4}$,} \\[4pt]
\frac{6\sqrt{3}(10\ell^4-81 \ell^3+206 \ell^2-171 \ell+36)}{(4 \ell^3-15 \ell^2+14 \ell -3)^{3/2}} & \text{if $\ell \equiv \pm 1 \pmod{4}$,} \\[4pt]
\frac{6\sqrt{3}(10\ell^4-81 \ell^3+206 \ell^2-156 \ell-24)}{(4 \ell^3-15 \ell^2+14 \ell)^{3/2}} & \text{if $\ell \equiv 2 \pmod{4}$.}
\end{cases}
\end{equation*}
\end{corollary}

\section{Computer-assisted calculation of kurtosis and fifth moment}\label{Curtis}

A computer program was used to find the fourth central moment of $\ssac$.
The program first finds representatives for each isomorphism class $\fC$ in $\Isom(4)$.
This is done by the method described at the end of \cref{Scott}, and the program finds $97$ isomorphism classes.
For each class $\fC$ in $\Isom(4)$, the program determines $|\fC|$ using an orbit-stabilizer technique and determines $\Sols(\fC,\ell)$ using Ehrhart theory and inclusion-exclusion, since finding $\Sols(\fC,\ell)$ requires one to count the number of integer solutions of a homogeneous linear system that lie in a hypercube as a function of the size of the hypercube (see \cite[Ch.\ 3]{Beck-Robins}) and to then deduct the number of solutions whose coordinates do not have distinct values.
The program uses these calculations to compute the sum in \cref{Sanria} with $p=4$, and thereby determines the fourth central moment of $\ssac$.
The result is given below as \cref{Fred}.
The program was written in C++ and employing the GNU Multiple Precision Arithmetic Library (GMP) \cite{GMP}, and obtained the fourth central moment of $\ssac$ in about $5$ seconds on a personal computer.
The same program also obtained the second and third moments of $\ssac$, and its results agree with our hand calculations in Sections \ref{Veronica} and \ref{Simon}.
With a few hours of computation time, the program was also able to find that $\Isom(5)$ has $2581$ isomorphism classes and then to compute an exact formula for the fifth central moment of $\ssac$ as a quasi-polynomial of degree $7$ and period $55440$.  

\begin{theorem}\label{Fred}
For $\ell \in \N$, the quantity $\mom{4}\ssac(f)$ is a quasi-polynomial function of $\ell$ of degree $6$ and period $120$ given by
\[
\mom{4} \ssac(f)=\frac{1}{45} \sum_{j=0}^6 a_j(\ell) \ell^j,
\]
where for every $\ell$ we have $a_6(\ell)=3840$; $a_5(\ell)=501120$; $a_4(\ell)=-6786480$;
\begin{align*}
a_3(\ell) & =
\begin{cases}
$27078080$ & \text{if $\ell \equiv 0 \pmod{2}$,}  \\
$27072320$ & \text{if $\ell \equiv 1 \pmod{2}$;}
\end{cases}
\\[4pt]
a_2(\ell) & =
\begin{cases}
$-17638464$ & \text{if $\ell \equiv 0 \pmod{2}$,}  \\
$-18213024$ & \text{if $\ell \equiv 1 \pmod{2}$;}
\end{cases}
\\[4pt]
a_1(\ell) & =
\begin{cases}
$-69561600$ & \text{if $\ell \equiv 0 \pmod{12}$,} \\
$-71342400$ & \text{if $\ell \equiv \pm 1, \pm 5 \pmod{12}$,} \\
$-75982080$ & \text{if $\ell \equiv \pm 2 \pmod{12}$,} \\
$-68516160$ & \text{if $\ell \equiv \pm 3 \pmod{12}$,} \\
$-72387840$ & \text{if $\ell \equiv \pm 4 \pmod{12}$,} \\
$-73155840$ & \text{if $\ell \equiv 6 \pmod{12}$;}
\end{cases}
\end{align*}
and $a_0(\ell)$ is a function of period $120$ whose values are given on \cref{Ernest}.
\end{theorem}
\begin{table}[ht!]
\begin{center}
\caption{Values of $a_0(\ell)$ as a function of $\ell \pmod{120}$}\label{Ernest}
\begin{tabular}{|l|r||l|r||l|r|}
\hline
$\ell \pmod{120}$       & $a_0(\ell)$ & $\ell \pmod{120}$       & $a_0(\ell)$ & $\ell \pmod{120}$       & $a_0(\ell)$ \\
\hline
\rowcolor{tableshade}
$  0$                      & $        0$ & $ 21$, $ 69$               & $ 53732304$ & $ 51$, $ 99$               & $ 57464784$ \\
$  1$, $ 49$               & $ 68764624$ & $ 22$, $ 58$, $ 82$, $118$ & $100980736$ & $ 53$, $ 77$               & $ 76964816$ \\
\rowcolor{tableshade}
$  2$, $ 38$, $ 62$, $ 98$ & $ 98195456$ & $ 23$, $ 47$               & $ 79591376$ & $ 55$                      & $ 60110800$ \\
$  3$, $ 27$               & $ 63657936$ & $ 24$, $ 96$               & $ 12386304$ & $ 56$, $104$               & $ 43065344$ \\
\rowcolor{tableshade}
$  4$, $ 76$               & $ 48062464$ & $ 25$                      & $ 56378320$ & $ 60$                      & $  2211840$ \\
$  5$                      & $ 58385360$ & $ 28$, $ 52$               & $ 54255616$ & $ 61$, $109$               & $ 69870544$ \\
\rowcolor{tableshade}
$  6$, $ 54$, $ 66$, $114$ & $ 61323264$ & $ 29$, $101$               & $ 70771664$ & $ 63$, $ 87$               & $ 62552016$ \\
$  7$, $103$               & $ 78690256$ & $ 30$, $ 90$               & $ 48936960$ & $ 65$                      & $ 57279440$ \\
\rowcolor{tableshade}
$  8$, $ 32$               & $ 49258496$ & $ 31$, $ 79$               & $ 72497104$ & $ 68$, $ 92$               & $ 51470336$ \\
$  9$, $ 81$               & $ 52626384$ & $ 33$, $ 57$               & $ 58819536$ & $ 71$, $119$               & $ 73398224$ \\
\rowcolor{tableshade}
$ 10$, $ 70$               & $ 82401280$ & $ 34$, $ 46$, $ 94$, $106$ & $ 94787584$ & $ 73$, $ 97$               & $ 74957776$ \\
$ 11$, $ 59$               & $ 74504144$ & $ 35$                      & $ 62117840$ & $ 75$                      & $ 45078480$ \\
\rowcolor{tableshade}
$ 12$, $108$               & $ 20791296$ & $ 36$, $ 84$               & $ 14598144$ & $ 80$                      & $ 30679040$ \\
$ 13$, $ 37$               & $ 76063696$ & $ 39$, $111$               & $ 56358864$ & $ 83$, $107$               & $ 80697296$ \\
\rowcolor{tableshade}
$ 14$, $ 26$, $ 74$, $ 86$ & $ 92002304$ & $ 40$                      & $ 33464320$ & $ 85$                      & $ 57484240$ \\
$ 15$                      & $ 43972560$ & $ 41$, $ 89$               & $ 69665744$ & $ 88$, $112$               & $ 52043776$ \\
\rowcolor{tableshade}
$ 16$, $ 64$               & $ 45850624$ & $ 43$, $ 67$               & $ 79796176$ & $ 93$, $117$               & $ 59925456$ \\
$ 17$, $113$               & $ 75858896$ & $ 44$, $116$               & $ 45277184$ & $ 95$                      & $ 61011920$ \\
\rowcolor{tableshade}
$ 18$, $ 42$, $ 78$, $102$ & $ 67516416$ & $ 45$                      & $ 41346000$ & $100$                      & $ 35676160$ \\
$ 19$, $ 91$               & $ 73603024$ & $ 48$, $ 72$               & $ 18579456$ & $105$                      & $ 40240080$ \\
\rowcolor{tableshade}
$ 20$                      & $ 32890880$ & $ 50$, $110$               & $ 79616000$ & $115$                      & $ 61216720$ \\
\hline
\end{tabular}
\end{center}
\end{table}
Since $\ADF(f)=-1+\ssac(f)/\ell^2$, we can divide this result by $\ell^8$ to obtain the fourth central moment of the demerit factor.
\begin{corollary}\label{Geoffrey}
If $\ell \in \Z_+$, then 
\[
\mom{4} \ADF(f) =\frac{\mom{4}\ssac(f)}{\ell^8},
\]
where $\mom{4}\ssac(f)$ is the quasi-polynomial function of degree $6$ and period $120$ described in \cref{Fred}.
\end{corollary}
We can normalize the fourth central moment using our variance calculation from \cref{Julie} to obtain the kurtosis of $\ssac(f)$, which is the same as the kurtosis of $\ADF(f)=-1+\ssac(f)/\ell^2$.
\begin{corollary}\label{Curt}
If $\ell \in \Z_+$, then
\[
\smom{4} \ADF(f)=\smom{4}\ssac(f)=\frac{\mom{4} \ssac(f)}{\left(\mom{2} \ssac(f)\right)^2},
\]
where $\mom{4} \ssac(f)$ is the quasi-polynomial function of degree $6$ and period $120$ described in \cref{Fred} and $\mom{2} \ssac(f)$ is the quasi-polynomial function of degree $3$ and period $2$ described in \cref{Julie}.
\end{corollary}

\appendix

\section{Lemmas used to prove \cref{Sanri}}\label{Paul}

This section contains the technical results that we used in the proof of \cref{Sanri} in \cref{Nellie}.
\begin{lemma}\label{Nestor}
Let $R$ be a commutative ring, let $k \in \N$, let $E_1,\ldots,E_k$ be pairwise disjoint subsets of $\N$, let $E=\bigcup_{j=1}^k E_k$, and let $g_j\colon \As(E_j,=,\ell) \to R$ for each $j \in \{1,\ldots,k\}$.
Then
\[
\sum_{\tau \in \As(E,=,\ell)} \,\,\, \prod_{j=1}^k g_j(\tau\vert_{E_j}) = \prod_{j=1}^k \,\,\, \sum_{\tau^{(j)} \in \As(E_j,=,\ell)}  g_j(\tau^{(j)}).
\]
\end{lemma}
\begin{proof}
First, we let $\phi \colon \As(E,=,\ell) \to \prod_{j=1}^k \As(E_j,=,\ell)$ be the map with $\phi(\tau)=(\tau\vert_{E_1},\ldots,\tau\vert_{E_k})$ and let $\psi\colon \prod_{j=1}^k \As(E_j,=,\ell) \to \As(E,=,\ell)$ with $\psi((\tau^{(1)},\ldots,\tau^{(k)}))$ the element $\tau \in \As(E,=,\ell)$ such that for any $(e,s,v) \in E\times\lindexset$, we let $\tau_{e,s,v}=\tau^{(j)}_{e,s,v}$ where $j$ is the unique element of $\{1,\ldots,k\}$ with $e \in E_j$.
It is easy to see that $\phi$ and $\psi$ are inverse maps, and so they are bijective.
Now let $g\colon \prod_{j=1}^k \As(E_j,=,\ell) \to R$ with $g((\tau^{(1)},\ldots,\tau^{(k)})) = \prod_{j=1}^k g_j(\tau^{(j)})$.
Then we have
\begin{align*}
\prod_{j=1}^k \,\,\, \sum_{\tau^{(j)} \in \As(E_j,=,\ell)}  g_j(\tau^{(j)})
& = \sum_{(\tau^{(1)},\ldots,\tau^{(k)}) \in \prod_{j=1}^k \As(E_j,=,\ell)} \,\,\, \prod_{j=1}^k g_j(\tau^{(j)}) \\
& = \sum_{(\tau^{(1)},\ldots,\tau^{(k)}) \in \prod_{j=1}^k \As(E_j,=,\ell)} g((\tau^{(1)}, \cdots, \tau^{(k)})) \\
& = \sum_{\tau \in \As(E,=,\ell)} g(\phi(\tau)) \\
& = \sum_{\tau \in \As(E,=,\ell)} g((\tau\vert_{E_1},\ldots,\tau\vert_{E_k})) \\
& = \sum_{\tau \in \As(E,=,\ell)} \,\,\, \prod_{j=1}^k g_j(\tau\vert_{E_j}),
\end{align*}
where the first equality is due to the distributive law, the second and fifth are by the definition of $g$, the third is by bijectivity of $\phi$, and the fourth by the definition of $\phi$.
\end{proof}
For the rest of this section, we use the notation that if $\tau \in \As(E)$ for some $E\subseteq \N$ and $f=(\ldots,f_0,f_1,f_2,\ldots)$ is a sequence of elements of some commutative ring, then $f^\tau= \prod_{\gamma \in \Eindexset } f_{\tau_\gamma}$.
\begin{corollary}\label{Valerian}
Let $R$ be a commutative ring, let $k \in \N$, let $E_1,\ldots,E_k$ be pairwise disjoint subsets of $\N$, let $E=\bigcup_{j=1}^k E_k$, and let $f=(\ldots,f_0,f_1,f_2,\ldots)$ be a sequence of elements of $R$.
Then
\[
\sum_{\tau \in \As(E,=,\ell)} f^\tau = \prod_{j=1}^k \sum_{\tau^{(j)} \in \As(E_j,=,\ell)}  f^{\tau^{(j)}}.
\]
\end{corollary}
\begin{proof}
For each $j \in \{1,\ldots,k\}$, let $g_j \colon \As(E_j,=,\ell) \to R$ with $g_j(\upsilon)=f^\upsilon$ in \cref{Nestor} to obtain
\begin{align*}
\prod_{j=1}^k \sum_{\tau^{(j)} \in \As(E_j,=,\ell)}  f^{\tau^{(j)}}
& = \sum_{\tau \in \As(E,=,\ell)} \prod_{j=1}^k f^{\tau\vert_{E_j}} \\
& = \sum_{\tau \in \As(E,=,\ell)} \prod_{j=1}^k \prod_{(e,s,v) \in E_j\times\lindexset} f_{\tau_{e,s,v}} \\
& = \sum_{\tau \in \As(E,=,\ell)} \prod_{(e,s,v) \in E\times\lindexset} f_{\tau_{e,s,v}} \\
& = \sum_{\tau \in \As(E,=,\ell)} f^{\tau}.\qedhere
\end{align*}
\end{proof}
For the next two results, recall our notation $\ev$ for expectation value as $f$ ranges over all binary sequences of length $\ell$ with uniform measure.
\begin{corollary}\label{Vladimir}
If $k,\ell \in \N$ and $E_1,\ldots,E_k$ are  pairwise disjoint subsets of $\N$ with $E=\bigcup_{j=1}^k E_k$, then
\[
\sum_{\tau \in \As(E,=,\ell)} \prod_{j=1}^k \ev(f^{\tau\vert_{E_j}}) = \prod_{j=1}^k \sum_{\tau^{(j)} \in \As(E_j,=,\ell)}  \ev(f^{\tau^{(j)}}).
\]
\end{corollary}
\begin{proof}
This is just the special case of \cref{Nestor} where for each $j \in \{1,\ldots,k\}$, we let $g_j \colon \As(E_j,=,\ell) \to \Q$ with $g_j(\upsilon)=\ev(f^\upsilon)$.
\end{proof}

\begin{lemma}\label{Kirkland} 
Let $p, \ell \in \N$, let $\cP\in \Part(p)$, let $\tau \in \As(\cP,=,\ell)$, and let $E \subseteq [p]$.  Then
\[
\ev(f^{\tau\res{E}})=\begin{cases}
1 & \text{if $\cP_E$ is even,} \\
0 & \text{otherwise.}
\end{cases} 
\]
\end{lemma}
\begin{proof}
Note that 
$f^{\tau\res{E}}= \prod_{\gamma \in \Eindexset} f_{(\tau\res{E})_\gamma} = \prod_{j \in [\ell]} f_j^{|(\tau\res{E})^{-1}(\{j\})|}$.
Since $f_0,\ldots,f_{\ell-1}$ are independent and uniformly distributed on $\{\pm 1\}$, we see that $\ev(f^{\tau\res{E}})$ is $1$ if $|(\tau\res{E})^{-1}(\{j\})|$ is even for all $j \in [\ell]$, and $\ev(f^{\tau\res{E}})=0$ otherwise.
\cref{flower}\ref{flower-h} shows that $\tau\res{E} \in \As(\cP_E,=,\ell)$, and so $\cP_E$ is the partition of $\Eindexset$ induced by $\tau\res{E}$, so the classes of $\cP_E$ are the nonempty fibers of $\tau\res{E}$.
Thus the cardinalities of the nonempty fibers of $\tau\res{E}$ are the cardinalities of the classes of $\cP_E$, and so $|(\tau\res{E})^{-1}(\{j\})|$ is even for every $j \in [\ell]$ if and only if $\cP_E$ is even.
\end{proof}

\section{Linear algebra lemmas}\label{Lester}
This section gives a general criterion for the homogeneous linear systems in \cref{Scott} to have solutions whose coordinates have distinct values.  Recall the Hamming weight and distance from \cref{Hamilcar}.
\begin{lemma}\label{Evelyn}
Let $F$ be a subfield of $\C$, let $M$ be a matrix with entries in $F$ whose row sums are all zero, and let $M'$ be the reduced row echelon form of $M$.  Then the homogeneous system $M x=0$ has a solution whose coordinates have distinct values in $F$ if and only if $M'$ has no row of Hamming weight $2$ nor any pair of rows with Hamming distance $2$.
\end{lemma}
\begin{proof}
Assume that $M'$ has no row of Hamming weight $2$ nor any pair of rows with Hamming distance $2$.  Since row sums are zero, this means that Hamming weights of rows and Hamming distances between rows cannot be $1$ or $2$.  Let $E$ be the set that consists of entries of $M'$ along with the elements $0$, $1$, and $-1$, and let $D=\{e-e':e,e' \in E\}$.  Let $s$ and $\ell$ respectively be the smallest and largest magnitudes of the nonzero elements of $D$.  For any solution $x$ to $M'x=0$, the difference between the values assigned to any two distinct variables (pivot or free) is expressible as an $F$-linear combination of the values of the free variables, with the coefficients in the combination being elements of $D$ and with at least one of these coefficients being nonzero due to the Hamming weight condition.  Let $q$ be a rational number (so an element of $F$) with $q > (c+1) \ell/s$, where $c$ is the number of columns in $M'$.  So clearly $q > 1$.  Let $k$ be the number of free variables, and consider the solution one obtains by assigning the powers $q,q^2,\ldots,q^k$ to all the free variables and letting the pivot variables take the values that they must (these will be in $F$ since $M'$ has entries in $F$ because $F$ is a field and $M$ has entries in $F$).  The difference between any two values in our solution is a sum of the form $d=\sum_{i=1}^j d_i q^i$ for some $j \in \{1,\ldots, k\}$ with $d_1,\ldots,d_j \in D$ and $d_j\not=0$.  But notice that $|d_j q^j| \geq s q^j > (c+1) \ell q^{j-1} > k \ell q^{j-1} \geq j \ell q^{j-1} > \sum_{i=1}^{j-1} |d_i| q^i \geq |\sum_{i=1}^{j-1} d_i q^i|$.  Thus $d\not=0$.

If the condition of Hamming weights and distances is not met, then some element $r$ in the row space of $M$ has Hamming weight $2$ and row sum zero, so $r$ must have one nonzero entry $u$ (say, in the $i$th coordinate) and another nonzero entry $-u$ (say, in the $j$th coordinate), and the rest of the entries are zero.  This implies that any solution $x$ to the linear system $M x=0$ must have identical values for its $i$th and $j$th coordinates.
\end{proof}
\begin{corollary}\label{Edith}
Let $M$ be a matrix with integer entries whose row sums are all zero and let $M'$ be the reduced row echelon form of $M$.  The homogeneous system $M x=0$ has a solution whose coordinates have distinct values in $\N$ if and only if $M'$ has no row of Hamming weight $2$ and no pair of rows with Hamming distance $2$.
\end{corollary}
\begin{proof}
The zero sum condition means that the set of solutions of $M$ is closed under translation by the all-one vector and by scaling, so a solution of the form we seek exists if and only if there is a solution whose coordinates have distinct rational values; then apply \cref{Evelyn} with $F=\Q$.
\end{proof}

\section{Combinatorial lemmas}\label{Arthur}

This appendix contains some results in enumerative combinatorics that are useful in calculating the cardinality of $\As(\cP,=,\ell)$ for various partitions $\cP \in \Con(p)$ for $p$ a positive integer and $\ell\in\N$.
These results are used in Sections \ref{Veronica} and \cref{Apple-proof}.

The first result, from \cite[Lemma 12]{Abrego-Fernandez-Katz-Kolesnikov}, counts the number of arithmetic progressions of a certain size within a larger arithmetic progression.
For our purposes, an $n$-term arithmetic progression is a set of the form $\{a,a+b,a+2 b,\ldots,a+(n-1)b\}$ for some integers $a,b,n$ with $b\not=0$.
\begin{lemma} \label{Fuzz}
Let $\ell \in \N$ and $k \in \Z_+$ with $\ell=q k+r$, where $q,r\in\Z$ and $0\leq r<k$.
Then exactly $(\ell-r)(\ell+r-k)/(2 k)$ subsets of $[\ell]$ are $(k+1)$-term arithmetic progressions.
\end{lemma}
The next two counting results follow from \cref{Fuzz}.
\begin{corollary}\label{Persephone}
Let $\ell \in \N$ and let $N$ be the number of $3$-term arithmetic progressions in $[\ell]$.
Then the following four numbers are all equal
\begin{enumerate}[label=(\roman*)]
\item $2 N$,
\item the number of ordered pairs $(A,B)$ where $A,B$ are distinct elements in $[\ell]$ such that $A \equiv B \pmod{2}$,
\item the number of solutions to the equation $A+B=2 C$ with $A,B,C$ distinct elements of $[\ell]$,
\item $\floor{(\ell-1)^2/2}$.
\end{enumerate}
The number of solutions to the equation $A+B=2 C$ with $A, B, C$  (not necessarily distinct) elements of $[\ell]$ is $\floor{(\ell^2+1)/2}$.
\end{corollary}

\begin{proof}
If $(A,B)$ is a pair of distinct elements of $[\ell]$, then there is some $C \in [\ell]$ such that $A+B=2 C$ if and only if $A$ and $B$ have the same parity, in which case $\{A,C,B\}$ is a $3$-term arithmetic progression.
This progression also arises if we begin with the pair $(B,A)$ (but in no other case), so the number of progressions is half the number of pairs, and one can check that our formula for the count is indeed twice what is given by the $k=2$ case of \cref{Fuzz}.

If we have a solution $(A, B, C)$ to $A+B=2 C$ where $A, B, C$ are not all distinct, then they must all be equal to each other, so there are $\ell$ such solutions with $A, B, C \in [\ell]$.  When we add this to the number of solutions with $A, B, C$ distinct we get $\floor{(\ell^2+1)/2}$.
\end{proof}

\begin{lemma}\label{Light}
For $\ell \in \N $, the number of choices of distinct $ A,B,C,D,E \in [\ell] = \{ 0,1, \dots, \ell-1\} $ that will make (i) $B-A=C-B=E-D$, (ii) $A<B<C$, and (iii) $ D<E $ is 
\begin{align*}
\begin{cases}
\frac{5\ell^3-32\ell^2+52\ell}{24} & \text{if $\ell \equiv 0 \pmod{6}$,} \\[4pt]
\frac{5\ell^3-32\ell^2+55\ell-28}{24} & \text{if $\ell \equiv \pm 1 \pmod{6}$,} \\[4pt]
\frac{5\ell^3-32\ell^2+52\ell-16}{24} & \text{if $\ell \equiv \pm 2 \pmod{6}$,} \\[4pt]
\frac{5\ell^3-32\ell^2+55\ell-12}{24} & \text{if $\ell \equiv 3 \pmod{6}$.}
\end{cases}
\end{align*}
\end{lemma}

\begin{proof}
If $A$, $B$, $C$, $D$, $E$ satisfy our conditions, then we let $s$ be the common value of $B-A$, $C-B$, and $E-D$, so $s$ is some integer with $1\leq s \leq\floor{\ell/2}$ so that $A,B,C$ are distinct and lie within $[\ell]$ (and also $D,E$ are distinct).
For a given value of $s$, there are $\ell-2 s$ choices of $A,B,C$ and $\ell-s$ choices of $D,E$, where we are not yet insisting that $\{A,B,C\} \cap \{D,E\}=\emptyset$; the number of total ways to achieve this is 
\begin{align}
\begin{split}\label{Tybalt}
\sum_{s=1}^{\floor{\ell/2}} (\ell-2s)(\ell-s) = 			
\begin{cases}
\frac{5\ell^3-12\ell^2+4\ell}{24} & \text{if $\ell$ is even,} \\[4pt]
\frac{5\ell^3-12\ell^2+7\ell}{24} & \text{if $\ell$ is odd.}
\end{cases}
\end{split}
\end{align}
Now we need to deduct the number of cases where $\{A,B,C\} \cap \{D,E\}\not=\emptyset$, which happens in four ways: (1) $A=D$ and $B=E$, (2) $B=D$ and $C=E$, (3) $A=E$, or (4) $C=D$.
In each of cases (1) and (2) $\{A,B,C,D,E\}$ is a $3$-term arithmetic progression in $[\ell]$, so by \cref{Persephone} there are $\floor{(\ell-1)^2/2} /2$ ways for this to occur.
In each of cases (3) and (4) $\{A,B,C,D,E\}$ is a $4$-term arithmetic progression in $[\ell]$, so one can use \cref{Fuzz} to see that there are $\floor{(\ell-1)(\ell-2)/3}/2$ such progressions in $[\ell]$.
Deducting these four counts from \eqref{Tybalt} yields the desired result.
\end{proof}
The following two results are easy exercises in counting.
\begin{lemma}\label{Sam}
For $\ell,r \in \N$, we have
\begin{enumerate}[label=(\roman*)]
\item\label{Wilfred} $\displaystyle\sum_{0 \leq j \leq \ell} j^2 = \frac{\ell(\ell+1)(2\ell+1)}{6}$ and
\item\label{Xavier} $\displaystyle\sums{0 \leq j \leq \ell \\ j\equiv r \!\!\pmod{2} } j^2 = \begin{cases}
\frac{\ell(\ell+1)(\ell+2)}{6} & \text{if $\ell \equiv r \pmod{2}$,} \\[4pt]
\frac{(\ell-1)\ell(\ell+1)}{6} & \text{if $\ell \not\equiv r \pmod{2}$.}
\end{cases}$
\end{enumerate}
\end{lemma}

\begin{lemma}\label{wow} 
Let $k,\ell \in \N$.  Then the number of ordered pairs $(A,B)$ of (not necessarily distinct) elements in $[\ell]$ with $A+B=k$ is
\begin{align*}
\sums{A, B \in [\ell] \\ A+B = k } 1 =
\begin{cases}
k+1 & \text{if $0 \leq k \leq \ell-1$,} \\
2(\ell-1)-k+1 & \text{if $\ell-1 \leq k \leq 2(\ell-1)$,} \\
0 & \text{otherwise.}
\end{cases}
\end{align*}
\end{lemma}

This leads to our final result.
\begin{lemma}\label{wowzers}
Let $\ell \in \N $.
\begin{enumerate}[label=(\roman*)]
\item Then the number $(A,B,C,D) \in [\ell]^4$ with $A+B=C+D$ is $(2 \ell^3+\ell)/3$.
\item\label{wowzers-two} The number of $(A,B,C,D) \in [\ell]^4$ with $A$, $B$, $C$, and $D$ distinct and $A+B=C+D$ is
\[
\begin{cases}
\frac{2 \ell^3-9 \ell^2+10\ell}{3} & \text{if $\ell$ is even,} \\[4pt]
\frac{2 \ell^3-9 \ell^2+10\ell-3}{3} & \text{if $\ell$ is odd.}
\end{cases}
\]
\item The number of $(A,B,C,D) \in [\ell]^4$ where $A+B=C+D$ and $A \equiv B \pmod{2}$ is
\[
\begin{cases}
\frac{\ell^3-\ell}{3} & \text{if $\ell$ is even,} \\[4pt]
\frac{\ell^3+2\ell}{3} & \text{if $\ell$ is odd.}
\end{cases}
\]
\item\label{wowzers-four} The number of $(A,B,C,D) \in [\ell]^4$ with $A$, $B$, $C$, and $D$ distinct and $A+B=C+D$ and $A \equiv B \pmod{2}$ is
\[
\begin{cases}
\frac{\ell^3-6\ell^2+8\ell}{3} &  \text{if $\ell$ is even,} \\[4pt]
\frac{\ell^3-6\ell^2+11\ell -6}{3} & \text{if $\ell$ is odd.}
\end{cases}
\]
\end{enumerate}
\end{lemma}
\begin{proof}
The first count equals \[\sum_{0 \leq k \leq 2\ell-2} \left|\left\{(A,B) \in [\ell]^2:  A+B=k\right\}\right| \left|\left\{(C,D) \in [\ell]^2:  C+D=k\right\}\right|,\]
and \cref{wow} gives us the cardinalities of the sets, so we have a sum of squares $1^2+2^2+\cdots+(\ell-1)^2+\ell^2+(\ell-1)^2+\cdots+2^2+1^2$, which we total up using \cref{Sam}\ref{Wilfred} to get the final result.

To get the second count, we deduct solutions in which $A$, $B$, $C$, and $D$ are not distinct, which fall into five cases:
\begin{itemize}
\item The number of solutions when $A=B$ and $A$, $C$, and $D$ are distinct (or when $C=D$ and $A$, $B$, and $C$ are distinct) is given by \cref{Persephone} as $\floor{(\ell-1)^2/2}$.
\item The number of solutions when $A=C\not=B=D$ (or when $A=D\not=B=C$) is clearly $\ell(\ell-1)$.
\item The number of solutions when $A=B=C=D$ is clearly $\ell$.
\end{itemize}

The third count equals \[\sums{0 \leq k \leq 2\ell-2 \\ k \equiv 0 \!\!\!\!\pmod{2}} \left|\left\{(A,B) \in [\ell]^2:  A+B=k\right\}\right| \left|\left\{(C,D) \in [\ell]^2:  C+D=k\right\}\right|,\]
and \cref{wow} gives us the cardinalities of the sets, so we have a sum of squares $1^2+3^2+\cdots+(\ell-1)^2+(\ell-1)^2+\cdots+3^2+1^2$ (if $\ell$ is even) or $1^2+3^2+\cdots+(\ell-2)^2+\ell^2+(\ell-2)^2+\cdots+3^2+1^2$ (if $\ell$ is odd) which we total up using \cref{Sam}\ref{Xavier} to get the final result.

To get the fourth count, we deduct solutions in which $A$, $B$, $C$, and $D$ are not distinct, which fall into five cases:
\begin{itemize}
\item The number of solutions when $A=B$ and $A$, $C$, and $D$ are distinct (or when $C=D$ and $A$, $B$, and $C$ are distinct) is given by \cref{Persephone} as $\floor{(\ell-1)^2/2}$.
\item The number of solutions when $A \equiv B \pmod{2}$ and $A=C\not=B=D$ (or when $A\equiv B \pmod{2}$ and $A=D\not=B=C$) is $\floor{(\ell-1)^2/2}$ by \cref{Persephone}.
\item The number of solutions when $A=B=C=D$ is clearly $\ell$. \qedhere
\end{itemize}
\end{proof}

\section{Proof of \cref{Isabel}}\label{Isabel-proof}

To find a set of representatives for $\Isom(3)$, we use the procedure outlined at the end of \cref{Scott}.
We start by finding representatives modulo row and column permutations of all the absolute matrices for partitions in $\Con(3)$.
\cref{Clarence-abs} tells us that absolute matrices for partitions in $\Con(3)$ are $3\times t$ matrices with $4 \leq t \leq 6$ that have entries in $\{0,1,2\}$ where each column sum is a positive even integer and each row has either a $2$ and two $1$'s (and the remaining entries equal to $0$) or else four $1$'s (and any remaining entries equal to $0$).

We begin with the $3\times 4$ absolute matrices.  Since all the row sums are $4$, the average column sum must be $3$, but since column sums are even, this means that some column sum must be $4$ or greater.  The only way to get a column sum equal to $6$ is
\[
A= \begin{pmatrix}
2 & 0 & 1 & 1 \\
2 & 1 & 0 & 1 \\
2 & 1 & 1 & 0
\end{pmatrix},
\]
If there is no column sum equal to $6$, then the column sums must be $4$'s and $2$'s, and there must be two sums equal to $4$ and two sums equal to $2$.
A column with sum $4$ must have $2,2,0$ or $2,1,1$ in some order, but we can have at most one column with $2,2,0$, since we can only have at most one $2$ per row.
If we have one of each type of column with sum $4$, then we get
\[
B= \begin{pmatrix}
2 & 1 & 1 & 0 \\
2 & 1 & 0 & 1 \\
0 & 2 & 1 & 1
\end{pmatrix},
\]
but if both columns have $2,1,1$ in some order, then then we have
\[
C= \begin{pmatrix}
2 & 1 & 0 & 1 \\
1 & 2 & 0 & 1 \\
1 & 1 & 2 & 0
\end{pmatrix}
\quad \text{or} \quad
D= \begin{pmatrix}
2 & 1 & 1 & 0 \\
1 & 2 & 0 & 1 \\
1 & 1 & 1 & 1
\end{pmatrix},
\]
depending on whether all three rows have a $2$ or not.

Now we consider $3\times 5$ absolute matrices.
Since all row sums are $4$, the entries of the matrix add up to $12$, and since each column sum is a positive even integer, the column sums are $4,2,2,2,2$.
A column with sum $4$ must have $2,2,0$ or $2,1,1$ in some order.
In the former case, we get
\[
E= \begin{pmatrix}
2 & 0 & 0 & 1 & 1 \\
2 & 0 & 1 & 0 & 1 \\
0 & 2 & 1 & 1 & 0
\end{pmatrix}
\quad \text{or} \quad
F= \begin{pmatrix}
2 & 1 & 1 & 0 & 0 \\
2 & 0 & 0 & 1 & 1 \\
0 & 1 & 1 & 1 & 1
\end{pmatrix},
\]
depending on whether all three rows have a $2$ or not.
Otherwise, the column with sum $4$ must have $2,1,1$ in some order, and then we get
\[
G= \begin{pmatrix}
2 & 0 & 0 & 1 & 1 \\
1 & 2 & 0 & 1 & 0 \\
1 & 0 & 2 & 0 & 1
\end{pmatrix},
\quad 
H= \begin{pmatrix}
2 & 0 & 1 & 1 & 0 \\
1 & 2 & 0 & 0 & 1 \\
1 & 0 & 1 & 1 & 1
\end{pmatrix},
\quad \text{or} \quad 
I= \begin{pmatrix}
2 & 1 & 1 & 0 & 0 \\
1 & 1 & 0 & 1 & 1 \\
1 & 0 & 1 & 1 & 1
\end{pmatrix},
\]
depending on how many rows have a $2$.

Now we consider $3\times 6$ absolute matrices.  Since the sum of all the entries is $12$ and each column sum is a positive even integer, every column sum must be $2$.
We get
\begin{center}
\begin{tabular}{ll}
$J = \begin{pmatrix}
2 & 0 & 0 & 0 & 1 & 1 \\
0 & 2 & 0 & 1 & 0 & 1 \\
0 & 0 & 2 & 1 & 1 & 0
\end{pmatrix}$,
&
$K = \begin{pmatrix}
2 & 0 & 1 & 1 & 0 & 0 \\
0 & 2 & 0 & 0 & 1 & 1 \\
0 & 0 & 1 & 1 & 1 & 1
\end{pmatrix}$,
\\[20pt]
$L = \begin{pmatrix}
2 & 1 & 1 & 0 & 0 & 0 \\
0 & 1 & 0 & 1 & 1 & 1 \\
0 & 0 & 1 & 1 & 1 & 1
\end{pmatrix}$, or
& 
$M = \begin{pmatrix}
1 & 1 & 1 & 1 & 0 & 0 \\
1 & 1 & 0 & 0 & 1 & 1 \\
0 & 0 & 1 & 1 & 1 & 1
\end{pmatrix}$,
\end{tabular}
\end{center}
depending on how many rows have a $2$ in them.

Now we find representatives of classes of monochrome matrices that map to these by $\abs$, so we are assigning signs to the elements of the matrices so as to produce rows meeting condition \ref{Rowena-mono} of \cref{Clarence-mono}, and we are looking for representatives modulo row and column permutations and negation of any selection of rows.    If $X$ is an absolute matrix, then we label the monochrome matrices that arise from assigning signs to $X$ by $X_1,\ldots,X_n$ for some $n$, and in this way we obtain
\begin{center}
\begin{tabular}{ll}
$A_1= \begin{pmatrix}
2 &  0 & -1 & -1 \\
2 & -1 &  0 & -1 \\
2 & -1 & -1 &  0
\end{pmatrix}$,
&
$B_1= \begin{pmatrix}
2 & -1 & -1 &  0 \\
2 & -1 &  0 & -1 \\
0 &  2 & -1 & -1
\end{pmatrix}$,
\\[20pt]
$C_1= \begin{pmatrix}
 2 & -1 & 0 & -1 \\
-1 &  2 & 0 & -1 \\
-1 & -1 & 2 &  0
\end{pmatrix}$,
&
$D_1= \begin{pmatrix}
 2 & -1 & -1 &  0 \\
-1 &  2 &  0 & -1 \\
 1 &  1 & -1 & -1
\end{pmatrix}$,
\\[20pt]
$D_2= \begin{pmatrix}
 2 & -1 & -1 &  0 \\
-1 &  2 &  0 & -1 \\
 1 & -1 &  1 & -1
\end{pmatrix}$,
&
$D_3= \begin{pmatrix}
 2 & -1 & -1 &  0 \\
-1 &  2 &  0 & -1 \\
 1 & -1 & -1 &  1
\end{pmatrix}$,
\\[20pt]
$E_1= \begin{pmatrix}
2 & 0 &  0 & -1 & -1 \\
2 & 0 & -1 &  0 & -1 \\
0 & 2 & -1 & -1 &  0
\end{pmatrix}$,
&
$F_1= \begin{pmatrix}
2 & -1 & -1 &  0 &  0 \\
2 &  0 &  0 & -1 & -1 \\
0 &  1 &  1 & -1 & -1
\end{pmatrix}$,
\\[20pt]
$F_2= \begin{pmatrix}
2 & -1 & -1 &  0 &  0 \\
2 &  0 &  0 & -1 & -1 \\
0 &  1 & -1 &  1 & -1
\end{pmatrix}$,
&
$G_1= \begin{pmatrix}
 2 & 0 & 0 & -1 & -1 \\
-1 & 2 & 0 & -1 &  0 \\
-1 & 0 & 2 &  0 & -1
\end{pmatrix}$,
\\[20pt]
$H_1= \begin{pmatrix}
 2 & 0 & -1 & -1 &  0 \\
-1 & 2 &  0 &  0 & -1 \\
 1 & 0 &  1 & -1 & -1
\end{pmatrix}$,
&
$H_2= \begin{pmatrix}
 2 & 0 & -1 & -1 &  0 \\
-1 & 2 &  0 &  0 & -1 \\
 1 & 0 & -1 & -1 &  1
\end{pmatrix}$,
\\[20pt]
$I_1= \begin{pmatrix}
2 & -1 & -1 &  0 &  0 \\
1 &  1 &  0 & -1 & -1 \\
1 &  0 &  1 & -1 & -1
\end{pmatrix}$,
&
$I_2= \begin{pmatrix}
2 & -1 & -1 &  0 &  0 \\
1 &  1 &  0 & -1 & -1 \\
1 &  0 & -1 &  1 & -1
\end{pmatrix}$,
\\[20pt]
$I_3= \begin{pmatrix}
2 & -1 & -1 &  0 &  0 \\
1 & -1 &  0 &  1 & -1 \\
1 &  0 &  1 & -1 & -1
\end{pmatrix}$,
&
$I_4= \begin{pmatrix}
2 & -1 & -1 &  0 &  0 \\
1 & -1 &  0 &  1 & -1 \\
1 &  0 & -1 &  1 & -1
\end{pmatrix}$,
\\[20pt]
$I_5= \begin{pmatrix}
2 & -1 & -1 &  0 &  0 \\
1 & -1 &  0 &  1 & -1 \\
1 &  0 & -1 & -1 &  1
\end{pmatrix}$,
&
$J_1= \begin{pmatrix}
2 & 0 & 0 &  0 & -1 & -1 \\
0 & 2 & 0 & -1 &  0 & -1 \\
0 & 0 & 2 & -1 & -1 &  0
\end{pmatrix}$,
\\[20pt]
$K_1= \begin{pmatrix}
2 & 0 & -1 & -1 &  0 &  0 \\
0 & 2 &  0 &  0 & -1 & -1 \\
0 & 0 &  1 &  1 &  -1 & -1
\end{pmatrix}$,
&
$K_2= \begin{pmatrix}
2 & 0 & -1 & -1 &  0 &  0 \\
0 & 2 &  0 &  0 & -1 & -1 \\
0 & 0 &  1 & -1 &  1 & -1
\end{pmatrix}$,
\end{tabular}
\end{center}
\begin{center}
\begin{tabular}{ll}
$L_1= \begin{pmatrix}
2 & -1 & -1 & 0 &  0 &  0 \\
0 &  1 &  0 & 1 & -1 & -1 \\
0 &  0 &  1 & 1 & -1 & -1
\end{pmatrix}$,
&
$L_2= \begin{pmatrix}
2 & -1 & -1 &  0 &  0 &  0 \\
0 &  1 &  0 &  1 & -1 & -1 \\
0 &  0 &  1 & -1 &  1 & -1
\end{pmatrix}$,
\\[20pt]
$M_1= \begin{pmatrix}
 1 &  1 & -1 & -1 &  0 &  0 \\
-1 & -1 &  0 &  0 &  1 &  1 \\
 0 &  0 &  1 &  1 & -1 & -1
\end{pmatrix}$,
&
$M_2= \begin{pmatrix}
1 & 1 & -1 & -1 &  0 &  0 \\
1 & 1 &  0 &  0 & -1 & -1 \\
0 & 0 &  1 & -1 &  1 & -1
\end{pmatrix}$,
\\[20pt]
$M_3= \begin{pmatrix}
1 &  1 & -1 & -1 &  0 &  0 \\
1 & -1 &  0 &  0 &  1 & -1 \\
0 &  0 &  1 & -1 &  1 & -1
\end{pmatrix}$,
&
$M_4= \begin{pmatrix}
1 & -1 &  1 & -1 &  0 &  0 \\
1 & -1 &  0 &  0 &  1 & -1 \\
0 &  0 &  1 & -1 &  1 & -1
\end{pmatrix}$, and
\\[20pt]
$M_5= \begin{pmatrix}
 1 & -1 & -1 &  1 &  0 &  0 \\
-1 &  1 &  0 &  0 &  1 & -1 \\
 0 &  0 &  1 & -1 & -1 &  1
\end{pmatrix}$.
\end{tabular}
\end{center}
When we compute the reduced row echelon forms of these matrices, we find that those arising from $A_1$, $B_1$, $C_1$, $D_2$, $D_3$, $E_1$, $F_2$, $H_2$, $I_2$, $I_3$, $I_4$, $L_2$, $M_2$, $M_3$, and $M_4$ all have rows with Hamming weight $2$, while those arising from $H_1$, $I_1$, $K_1$, and $L_1$ have pairs of rows with Hamming distance $2$, so these matrices do not satisfy condition \ref{Sanjay-mono} of \cref{Clarence-mono}.  On the other hand, matrices $D_1$, $F_1$, $G_1$, $I_5$, $J_1$, $K_2$, $M_1$, and $M_5$ do satisfy condition \ref{Sanjay-mono} of \cref{Clarence-mono}.
To obtain representatives of classes of display matrices that map to these monochrome matrices via $\mono$, we then color the $1$ and $-1$ entries of these eight matrices as we wish while respecting condition \ref{Rowena} of \cref{Clarence} to obtain
\begin{center}
\begin{tabular}{ll}
$D_1'= \begin{pmatrix}
 2 & -1_r & -1_b &  0 \\
-1_r &  2 &  0 & -1_b \\
 1_r &  1_b & -1_r & -1_b
\end{pmatrix}$,
&
$F_1'= \begin{pmatrix}
2 & -1_r & -1_b &  0 &  0 \\
2 &  0 &  0 & -1_r & -1_b \\
0 &  1_r &  1_b & -1_r & -1_b
\end{pmatrix}$,
\\[20pt]
$G_1'= \begin{pmatrix}
2 & 0 & 0 & -1_r & -1_b \\
-1_r & 2 & 0 & -1_b &  0 \\
-1_r & 0 & 2 &  0 & -1_b
\end{pmatrix}$,
&
$I_5'= \begin{pmatrix}
2 & -1_r & -1_b &  0 &  0 \\
1_r & -1_r &  0 &  1_b & -1_b \\
1_r &  0 & -1_r & -1_b &  1_b
\end{pmatrix}$,
\\[20pt]
$J_1'= \begin{pmatrix}
2 & 0 & 0 &  0   & -1_b & -1_r \\
0 & 2 & 0 & -1_r &  0   & -1_b \\
0 & 0 & 2 & -1_b & -1_r & 0
\end{pmatrix}$,
&
$K_2'= \begin{pmatrix}
2 & 0 & -1_r & -1_b &  0 &  0 \\
0 & 2 &  0 &  0 & -1_r & -1_b \\
0 & 0 &  1_r & -1_r &  1_b & -1_b
\end{pmatrix}$,
\\[20pt]
$M_1'= \begin{pmatrix}
 1_r &  1_b & -1_r & -1_b &  0 &  0 \\
-1_r & -1_b &  0 &  0 &  1_r &  1_b \\
 0 &  0 &  1_r &  1_b & -1_r & -1_b
\end{pmatrix}$, and
&
$M_5'= \begin{pmatrix}
 1_r & -1_r & -1_b &  1_b &  0 &  0 \\
-1_b &  1_b &  0 &  0 &  1_r & -1_r \\
 0 &  0 &  1_r & -1_r & -1_b &  1_b
\end{pmatrix}$.
\end{tabular}
\end{center}
The matrices $D_1'$, $F_1'$, $G_1'$, $I_5'$, $J_1'$, $K_2'$, $M_1'$, and $M_5'$ respectively represent the $\cP_1$, $\cP_2$, $\cP_3$, $\cP_4$, $\cP_5$, $\cP_6$, $\cP_7$, and $\cP_8$ described in the statement of this lemma.

\section{Proof of \cref{Genevieve}}\label{Genevieve-proof}

Since an isomorphism class of partitions in $\Con(3)$ is an orbit under the action of $\cWthree$, we use \cref{Stanley} to compute the size of the eight classes of $\Isom(3)$.
In particular, each class $\fC_j$ has a representative $\cP_j$ described in \cref{Isabel}, so $|\fC_j|=3\cdot 2^{10}/|\Stab(\cP_j)|$, where $\Stab(\cP_j)$ is the stabilizer of $\cP_j$ in $\cWthree$.
So proving the cardinalities in this lemma is tantamount to proving that $|\Stab(\cP_1)|=2^3$, $|\Stab(\cP_2)|=2^5$, $|\Stab(\cP_3)|=2^4$, $|\Stab(\cP_4)|=2^2$, $|\Stab(\cP_5)|=3\cdot 2^4$, $|\Stab(\cP_6)|=2^4$, $|\Stab(\cP_7)|=3\cdot 2^4$, and $|\Stab(\cP_8)|=3\cdot 2^2$.
Throughout this proof, we write an element $\pi\in\cWthree$ as
\[
\pi=\left(\left(\left((\digamma_{0,0},\digamma_{0,1}),\sigma_0\right),\ldots,\left((\digamma_{p-1,0},\digamma_{p-1,1}),\sigma_{p-1}\right)\right),\epsilon\right),
\]
where $\epsilon \in S_3$ and $\sigma_e$, and $\digamma_{e,s}$ are in $S_2$ for every $(e,s) \in [p]\times [2]$ and where
\[
\pi(e,s,v)=\left(\epsilon(e),\sigma_{\epsilon(e)}(s),\digamma_{\epsilon(e),\sigma_{\epsilon(e)}}(v)\right)
\]
for every $(e,s,v) \in \indexset$.
We write permutations in $S_2$ and $S_3$ using cycle notation, with $()$ representing the identity map.

One can check that $\pi$ stabilizes $\cP_1$ if and only if $\epsilon \in \{(),(01)\}$; $\sigma_0=\sigma_1=\sigma_2=()$; $\digamma_{01}=\digamma_{11}=()$; and $\digamma_{20}=\digamma_{21}=()$ if $\epsilon=()$ while $\digamma_{20}=\digamma_{21}=(01)$ if $\epsilon=(01)$.  Thus, we can make independent choices of $\epsilon \in \{(),(01)\}$ and $\digamma_{00}, \digamma_{10} \in S_2$, and then the rest are determined.  Thus, $|\Stab(\cP_1)|=2^3$.

One can check that $\pi$ stabilizes $\cP_2$ if and only if $\epsilon \in \{(),(01)\}$; $\sigma_0=\sigma_1=()$; $\sigma_2=()$ if $\epsilon=()$ but $\sigma_2=(01)$ if $\epsilon=(01)$; $\digamma_{01}=\digamma_{20}$; and $\digamma_{11}=\digamma_{21}$.  Thus, we can make independent choices of $\epsilon \in \{(),(01)\}$ and $\digamma_{00}, \digamma_{01}, \digamma_{10}, \digamma_{11} \in S_2$, and then the rest are determined.  Thus, $|\Stab(\cP_2)|=2^5$.

One can check that $\pi$ stabilizes $\cP_3$ if and only if $\epsilon \in \{(),(12)\}$; $\sigma_0=\sigma_1=\sigma_2=()$; and $\digamma_{11}=\digamma_{21}=()$; and $\digamma_{01}=()$ if $\epsilon=()$ but $\digamma_{01}=(01)$ if $\epsilon=(12)$.  Thus, we can make independent choices of $\epsilon \in \{(),(12)\}$ and $\digamma_{00}, \digamma_{10}, \digamma_{20} \in S_2$, and then the rest are determined.  Thus, $|\Stab(\cP_3)|=2^4$.

One can check that $\pi$ stabilizes $\cP_4$ if and only if $\epsilon \in \{(),(12)\}$; $\sigma_0=\sigma_1=\sigma_2=()$; $\digamma_{10}=\digamma_{11}=\digamma_{20}=\digamma_{21}=()$; and $\digamma_{01}=()$ if $\epsilon=()$ but $\digamma_{01}=(01)$ if $\epsilon=(12)$.  Thus, we can make independent choices of $\epsilon \in \{(),(12)\}$ and $\digamma_{00} \in S_2$ and then the rest are determined.  Thus, $|\Stab(\cP_4)|=2^2$.

One can check that $\pi$ stabilizes $\cP_5$ if and only if either (i) $\epsilon$ an even permutation and $\digamma_{01}=\digamma_{11}=\digamma_{21}=()$, or else (ii) $\epsilon$ is an odd permutation and $\digamma_{01}=\digamma_{11}=\digamma_{21}=(01)$, and we must have $\sigma_0=\sigma_1=\sigma_2=()$ in both cases.  Thus, we can make independent choices of $\epsilon \in S_3$ and $\digamma_{00}, \digamma_{10}, \digamma_{20}\in S_2$, and then the rest are determined.  Thus, $|\Stab(\cP_5)|=3! \cdot 2^3$.

One can check that $\pi$ stabilizes $\cP_6$ if and only if $\epsilon \in \{(),(01)\}$; $\sigma_0=\sigma_1=()$; $\digamma_{01}=\digamma_{11}=\sigma_2$; and $\digamma_{20}=\digamma_{21}=()$ if $\epsilon=()$ but $\digamma_{20}=\digamma_{21}=(01)$ if $\epsilon=(01)$.  Thus, we can make independent choices of $\epsilon \in \{(),(01)\}$ and $\sigma_2, \digamma_{00}, \digamma_{10} \in S_2$, and then the rest are determined.  Thus, $|\Stab(\cP_6)|=2^4$.

One can check that $\pi$ stabilizes $\cP_7$ if and only if either (i) $\epsilon$ is an even permutation and $\sigma_0=\sigma_1=\sigma_2=()$, or else (ii) $\epsilon$ is an odd permutation and $\sigma_0=\sigma_1=\sigma_2=(01)$, and we must have $\digamma_{00}=\digamma_{11}$, $\digamma_{10}=\digamma_{21}$ and $\digamma_{20}=\digamma_{01}$ in both cases.  Thus, we can make independent choices of $\epsilon \in S_3$ and $\digamma_{00}, \digamma_{10}, \digamma_{20} \in S_2$, and then the rest are determined.  Thus, $|\Stab(\cP_7)|=3! \cdot 2^3$.

One can check that $\pi$ stabilizes $\cP_8$ if and only if either (i) $\epsilon$ is an even permutation, $\sigma_0=\sigma_1=\sigma_2$, and $\digamma_{e,s}=()$ for every $(e,s) \in [3]\times[2]$; or else (ii) $\epsilon$ is an odd permutation, $\sigma_0=\sigma_1=\sigma_2$, and $\digamma_{e,s}=(01)$ for every $(e,s) \in [3]\times[2]$.  Thus, we can make independent choices of $\epsilon \in S_3$ and $\sigma_0 \in S_2$, and then the rest are determined.  Thus, $|\Stab(\cP_8)|=3! \cdot 2$.

\section{Proof of \cref{Apple}}\label{Apple-proof}

For each $j \in \{1,2,\ldots,8\}$, the partition $\cP_j$ as described in \cref{Isabel} is an element of $\fC_j$, so \cref{Sonia} tells us that $\Sols(\fC_j,\ell)=|\As(\cP_j,=,\ell)|$, which by \cref{Terrence} is the number of solutions whose coordinates have distinct values in $[\ell]$ of the homogeneous system $M x=0$ where $M$ is a monochrome matrix of $\cP_j$, examples of which can be found in the proof of \cref{Isabel} (in \cref{Isabel-proof}) as $D_1$, $F_1$, $G_1$, $I_1$, $J_1$, $K_2$, $M_1$, $M_5$ for $\cP_1,\ldots,\cP_8$ respectively.  So the rest of this proof consists of counting such solutions for the eight homogeneous linear systems derived from these matrices.

For $\Sols(\fC_1,\ell)$, we need to solve the system
\begin{align*}
2A &=B+C\\
2B &=A+D\\
A+B&=C+D
\end{align*}
with distinct $ A, B, C, D  \in [\ell]$.
These equations are equivalent to saying that $C, A, B, D$ is an arithmetic progression, and since the progression can be ascending or descending, we double the number of  of $4$-term arithmetic progressions residing in $[\ell]$, which is given by \cref{Fuzz}.

For $\Sols(\fC_2,\ell)$, we need to solve the system
\begin{align*}
2 A &= B+C \\
2 A &= D+E \\
B+C &= D+E
\end{align*}
with distinct $ A, B, C, D, E  \in [\ell]$.
Since $ A $ is the average of $ B $ and $ C $ and also the average of $ D $ and $ E $, then if $ B,C,D, E $  are distinct, then so are $ A, B, C, D, E $.
Hence, the number of solutions with distinct $ A, B, C, D, E  \in [\ell] $ to the above system equals the number of solutions with distinct $ B, C, D, E  \in [\ell] $ of 
\begin{equation*}
B+C = D+E
\end{equation*} 
such that $ B $ and $ C $ have the same parity. \cref{wowzers}\ref{wowzers-four} gives the number of solutions.

For $\Sols(\fC_3,\ell)$, we need to solve the system
\begin{align*}
2 A & = D + E \\
2 B & = A + D \\
2 C & = A + E
\end{align*}
with distinct $ A, B, C, D, E  \in [\ell]$.
This system is equivalent to saying that $E,C,A,B,D$ is an arithmetic progression, and since the progression can be ascending or descending, we double the number of  of $5$-term arithmetic progressions residing in $[\ell]$, which is given by \cref{Fuzz}.

For $\Sols(\fC_4,\ell)$, we need to solve the system
\begin{align*}
2 A & = B + C \\
A+D & = B + E \\
A+E & = C + D
\end{align*}
with distinct $ A, B, C, D, E  \in [\ell]$.
This system is equivalent to $E-D=A-B=C-A$, so that $A$ is the average of $B$ and $C$.
Hence, either (i) $C<A<B $ and $E<D $ or (ii) $B<A<C $ and $D<E$.
Any solution to (i) yields a solution to (ii) by interchanging the values of $C$ and $E$ with the values of $B$ and $D$, respectively.
So we restrict to (i), whose solutions are counted by \cref{Light}, and then double our counts.

For $\Sols(\fC_5,\ell)$, we need to solve the system
\begin{align*}
2 A & = E+F \\
2 B & = F+D \\
2 C & = D+E    
\end{align*}
with distinct $ A, B, C, D, E, F  \in [\ell]$.
Since $ A, B $, and $ C $ are averages, $D$, $E$, and $F$ must have the same parity.
Moreover, if $D$, $E$, and $F$ were distinct then $A$, $B$, and $C $ would be distinct from each other.
However, if $D$, $E$, and $F$ form an arithmetic progression then $|\{D,E,F\} \cap \{A,B,C\}| =1$, but if $D$, $E$, and $F$ are distinct and do not form an arithmetic progression, then all six values are distinct.
Hence, we must find how many ways we can pick $D$, $E$, and $F$ such that they have the same parity but do not form an arithmetic progression.
	
Now, there are $ 6 $ ways to order $D$, $E$, and $F$, and any solution to one ordering corresponds to a solution to a different ordering by interchanging values between variables.
So, for the rest of this paragraph, we restrict to the case $ D<E<F$, and afterward we shall multiply our count by $6$.
The number of ways to have $D$, $E$, and $F$ all be even is $ \binom{\ceil{\ell/2}}{3} $ and the number of ways to have $D$, $E$, and $F$ all be odd is $ \binom{\floor{\ell/2}}{3}$.
Hence the number of ways to have $D$, $E$, and $F$ such that they have the same parity is $\binom{\ceil{\ell/2}}{3} +  \binom{\floor{\ell/2}}{3}$.
Next, let $M$ be the number of integers of a particular parity in $ [\ell]$, which is $\ceil{\ell/2}$ (resp.\ $\floor{\ell/2}$) for even (resp.\ odd) parity.
By \cref{Fuzz}, if we write $M=2q+r$ with $q=\floor{M/2}$ and $0 \leq r <2 $, then the number of ways to have $D$, $E$, and $F$ with that particular parity in an arithmetic progression is $(M-r)(M+r-2)/4$.
So our solution depends on what $\ell $ is modulo $ 4 $.
If $\ell \equiv 0 \pmod 4 $, then our solution equals 
\[2\cdot \binom{\ell/2}{3} - 2 \cdot \frac{(\ell/2)(\ell/2-2)}{4} = \frac{\ell^3-9\ell^2+20\ell}{24}. \]
If $\ell \equiv 1  \pmod 4 $, then our count equals 
\[ \binom{(\ell+1)/2}{3}  + \binom{(\ell-1)/2}{3}  -  \frac{((\ell+1)/2-1)^2}{4} - \frac{((\ell-1)/2)((\ell-1)/2-2)}{4} = \frac{\ell^3-9\ell^2+23\ell-15}{24}. \]	
If $\ell \equiv 2 \pmod 4 $, then our count equals 
\[  2\cdot \binom{\ell/2}{3} - 2 \cdot \frac{(\ell/2-1)^2}{4} = \frac{\ell^3-9\ell^2+20\ell-12}{24}. \]
If $\ell \equiv 3 \pmod 4 $, then our count equals
\[ \binom{(\ell+1)/2}{3}  +\binom{(\ell-1)/2}{3}  -  \frac{((\ell+1)/2)((\ell+1)/2-2)}{4} - \frac{((\ell-1)/2-1)^2}{4}  = \frac{\ell^3-9\ell^2+23\ell-15}{24}.  \]  
We now multiply these counts by $6$ to get $|\Sols(\fC_5,\ell)|$.

For $\Sols(\fC_6,\ell)$, we need to solve the system
\begin{align*}
2 A & = C+D \\
2 B & = E+F \\
C+E & = D+F \\
\end{align*}
with distinct $ A,B,C,D,E,F \in [\ell] $.  This system is equivalent to saying that $C$, $A$, $D$ and $F$, $B$, $E$ are arithmetic progressions with the same spacings. Hence we either have (i) $ D<A<C $ and $ E<B<F $ or (ii) $ C<A<D $ and $ F<B<E $. Notice that any solution to (i) yields a solution to (ii) by interchanging the values of $ D $ and $ E $ with $ C $ and $ F $, respectively. So, we restrict to (i) and double our counts. 
		
Since $ D,A,C $ is an arithmetic progression in $ [\ell] $, then if we fix a spacing $ s $, there are $\ell-2s $ ways to have $ D<A<C $. Similarly, for a fixed spacing $ s $, there are $\ell-2s $ ways to have $ E<B<F $. Thus, if we allow $A,B,C,E,D,F\in [\ell] $ to not necessarily be distinct, then \cref{Sam}\ref{Xavier} tells us that the number of ways for (i) to occur is 
\[ \sum_{s=1}^{\floor{(\ell-1)/2}} (\ell-2s)^2 = \frac{\ell(\ell-1)(\ell-2)}{6}.\]
Since we are only counting positive spacings, then the only ways we can have $A,B,C,E,D,F$ not all distinct are when one of the three occurs: (a) $ \{D,A,C\} = \{E,B,F\}  $, (b) $ |\{D,A,C\} \cap \{E,B,F\}| = 2 $, or (c) $ |\{D,A,C\} \cap \{E,B,F\}| = 1 $. 

Now, the only way we can have (a) is if $ D=E $, $ A=B $, and $ C=F $. If $\ell \equiv r_2 \bmod 2$ with $ 0 \leq r_2 <2 $, then by \cref{Fuzz} there are $(\ell-r_2)(\ell+r_2-2)/4$ ways for (a) to occur. Next, for (b) to occur we must have either $ D=B $ and $ A=F $ or $ A=E $ and $ C=B $. By \cref{Fuzz}, if  $\ell \equiv r_3 \bmod 3$ with $ 0 \leq r_3 <3 $ there are $2(\ell-r_3)(\ell+r_3-3)/6$ ways for (b) to occur. Finally, for (c) to occur either $ D=F $ or $ C=E $. By \cref{Fuzz}, if  $\ell \equiv r_4 \bmod 4$ with $ 0 \leq r_4 <4 $ there are $2(\ell-r_4)(\ell+r_4-4)/8$ ways for (c) to occur. Thus, the number ways (i) can take place is
\[
\frac{\ell(\ell-1)(\ell-2)}{6} - \frac{(\ell-r_2)(\ell+r_2-2)}{4} - \frac{(\ell-r_3)(\ell+r_3-3)}{3}  - \frac{(\ell-r_4)(\ell+r_4-4)}{4}.
\]
To finish the proof, one works out what this expression becomes when we substitute for $r_2$, $r_3$, and $r_4$ the values that they take when $\ell$ runs through the congruence classes modulo $12$, and then remember to double the value to obtain $|\Sols(\fC_6,\ell)|$.

For $\Sols(\fC_7,\ell)$, we need to solve the system
\begin{align*}
A+B & = C+D \\
E+F & = A+B \\
C+D & = E+F
\end{align*}
with distinct $ A,B,C,D,E,F \in [\ell] $.  This system is equivalent to $A+B=C+D=E+F$. Let $ h = A+B $. In order for $A,B,C,D,E,F \in [\ell] $ to be distinct we need to be able to write $ h $ as the sum of two different numbers with at least three different unordered pairs of numbers. If $ h \leq \ell-1 $, then  by \cref{wow} there are $ h+1 $ ordered pairs $(x,y)$ with $x,y \in [\ell]$ such that $x+y=h$ and so there are $ \ceil{h/2} $ unordered pairs of numbers $\{x,y\} \subseteq [\ell]$ with $ x\not = y $ such that $ x+y=h $. If $ h>\ell-1 $ then, by \cref{wow} there are $ 2(\ell-1)-h+1 $ ordered pairs $(x,y)$ with $x,y \in [\ell]$ such that $x+y=h$ and so there are $\ceil{(2(\ell-1)-h)/2} $ unordered pairs of numbers $ \{x,y\} \subseteq [\ell]$ with $ x\not = y $ such that $ x+y=h $. Any three pairs that work can be assigned in $ 3! $ ways to the following sets $ \{ A,B\}, \{C,D\}, \{E,F\} $. Once all pairs are assigned, there are two different ways to assign numbers to variables for each $ 2 $-set. Hence, for a fixed $ h $ there are 
\begin{align*}
\begin{cases}
3!\cdot 2^3\cdot\binom{\ceil{h/2}}{3} & \text{if $h\leq \ell-1$,} \\[4pt]
3!\cdot 2^3\cdot\binom{\ceil{(2(\ell-1)-h)/2}}{3} & \text{if $h> \ell-1$}
\end{cases}
\end{align*}
ways to have $A=B=C+D=E+F=h$ with distinct $A,B,C,D,E,F \in [\ell]$.
Thus, the number of the solutions to the system must be 
\[
3! \cdot 2^3 \cdot \left[ \sum_{h=5}^{\ell-1} \binom{\ceil{h/2}}{3} + \sum_{h=\ell}^{2\ell-7}\binom{\ceil{(2(\ell-1)-h)/2}}{3}\right],
\]
which, when $\ell$ is even, is equal to
\[
3! \cdot 2^3 \cdot \left[ \binom{\ell/2}{3} + 4 \sum_{i=0}^{\ell/2-1} \binom{i}{3} \right] = 3! \cdot 2^3 \cdot \left[ \binom{\ell/2}{3} + 4 \binom{\ell/2}{4} \right],
\]
but when $\ell$ is odd, is equal to
\[
3! \cdot 2^3 \cdot \left[ -\binom{(\ell-1)/2}{3} + 4 \sum_{i=0}^{(\ell-1)/2} \binom{i}{3} \right] = 3! \cdot 2^3 \cdot \left[ -\binom{(\ell-1)/2}{3} + 4 \binom{(\ell+1)/2}{4} \right].
\]
These expressions for $\Sols(\fC_7,\ell)$ simplify to those in the statement of this lemma.

For $\Sols(\fC_8,\ell)$, we need to solve the system
\begin{align*}
A+D & = B+C \\
B+E & = A+F \\
C+F & = D+E
\end{align*}
with distinct $ A,B,C,D,E,F \in [\ell] $.
This system is equivalent to $A-B=C-D=E-F$.
Hence, we either have (i) $ A>B $, $ C>D $, $ E>F $ or (ii) $ B>A $, $ D>C $, $ F>E $. Since (i) yields a solution to (ii) by interchanging the values of $ A $, $ C $, and $ E $ with $ B $, $ D $, and $ F $, respectively.
So, we restrict to (i) and double our counts. 
	
Notice that if $ A,C,E $ are distinct from each other, then $ B,D,F $ are also distinct from each other.
Thus, we first count how many ways we can have $ A,C,E $ distinct from each other while $ A>B $. Now, if $ s = A-B$, then there are $\ell-s $ ways to have $ A>B $.
This leaves $\ell-s-1 $ options for $ C $ and $\ell-s-2 $ options for $ E $.
Thus, if we allow $ A,B,C,D,E,F \in [\ell] $ to not necessarily be distinct the number of ways (i) can occur is
\begin{equation}\label{pine}
\sum_{s=1}^{\ell-3} (\ell-s)(\ell-s-1)(\ell-s-2) = 3! \sum_{t=3}^{\ell-1} \binom{t}{3} = 3! \binom{\ell}{4} = \frac{\ell(\ell-1)(\ell-2)(\ell-3)}{4}.
\end{equation}
Since we are only counting with positive spacing $s$ between $A$ and $B$ while $ A,C,E $ are distinct from each other, then the only ways we can have $ A,B,C,D,E,F $ not all distinct are when either 
\begin{enumerate}[label=(\alph*).]
	\item $| \{A,B,D,C,E,F\} |=5 $ \, or
	\item $| \{A,B,D,C,E,F\} |=4 $.
\end{enumerate}
Now, the only way we can have (a) is if exactly one of $ |\{A\} \cap \{D,F\} | =1$, $ |\{C\} \cap \{B,F\} | =1$, or $ |\{E\} \cap \{B,D\} | =1$ occurs.
So, there are $ 6 $ different ways for (a) to occur.
Any case of (a) is an instance of  \cref{Light}. Hence, we must subtract $ 6 $ times 
\begin{align}
\begin{split}\label{spruce}
\begin{cases}
\frac{5\ell^3-32\ell^2+52\ell}{24} & \text{if $\ell \equiv 0 \pmod{6}$,} \\[4pt]
\frac{5\ell^3-32\ell^2+55\ell-28}{24} & \text{if $\ell \equiv \pm 1 \pmod{6}$,} \\[4pt]
\frac{5\ell^3-32\ell^2+52\ell-16}{24} & \text{if $\ell \equiv \pm 2 \pmod{6}$,} \\[4pt]
\frac{5\ell^3-32\ell^2+55\ell-12}{24} & \text{if $\ell \equiv 3 \pmod{6}$}
\end{cases}
\end{split}
\end{align}
from \eqref{pine}.
Next, the only way we can have (b) is if exactly one of $ (A,B) \in \{(D,E), (F,C)\}$, $(C,D) \in \{(B,E),(F,A) \}$, or $(E,F) \in \{(D,A),(B,C)\}$ occurs.
Hence, there are $ 6 $ different ways for (b) to occur.
Any case of (b) makes $\{A,B,C,D,E,F\}$ a $4$-term arithmetic progression, so it falls under \cref{Fuzz} with $ k=3 $. Thus, we must also subtract $ 6 $ times 
\begin{align}
\begin{split}\label{sequoia}
\begin{cases}
\frac{\ell^2-3\ell}{6} & \text{if $\ell \equiv 0 \pmod{3}$,} \\[4pt]
\frac{\ell^2 -3\ell + 2}{6} & \text{if $\ell \equiv \pm 1 \pmod{3}$}
\end{cases}
\end{split}
\end{align}
from \eqref{pine}.
When we have subtracted $6$ times \eqref{spruce} and $6$ times \eqref{sequoia} from \eqref{pine}, we get the count for case (i).  Doubling this to get the full count gives $|\Sols(\fC_8,\ell)|$.

\section*{Acknowledgement}

The authors thank Bernardo \'Abrego, Silvia Fern\'andez-Merchant, and the anonymous reviewers for helpful discussions and suggestions.


\begin{thebibliography}{AFMKK16}

\bibitem[AFMKK16]{Abrego-Fernandez-Katz-Kolesnikov}
Bernardo~M. \'Abrego, Silvia Fern\'andez-Merchant, Daniel~J. Katz, and Levon
  Kolesnikov.
\newblock On the number of similar instances of a pattern in a finite set.
\newblock {\em Electron. J. Combin.}, 23(4):Paper 4.39, 24, 2016.

\bibitem[ALS04]{Aupetit}
Sebastien Aupetit, Pierre Liardet, and Mohamed Slimane.
\newblock Evolutionary search for binary strings with low aperiodic
  auto-correlations.
\newblock In P.~Liardet, P.~Collet, C.~Fonlupt, E.~Lutton, and M.~Schoenauer,
  editors, {\em Artificial Evolution}, volume 2936 of {\em Lecture Notes in
  Computer Science}, pages 39--50, 2004.

\bibitem[BL01]{Borwein-Lockhart}
Peter Borwein and Richard Lockhart.
\newblock The expected {$L_p$} norm of random polynomials.
\newblock {\em Proc. Amer. Math. Soc.}, 129(5):1463--1472, 2001.

\bibitem[BR15]{Beck-Robins}
Matthias Beck and Sinai Robins.
\newblock {\em Computing the continuous discretely}.
\newblock Undergraduate Texts in Mathematics. Springer, New York, second
  edition, 2015.

\bibitem[{Fre}20]{GMP}
{Free Software Foundation, Inc.}
\newblock {\em {GNU} {MP} version {\tt 6.2.1}}, 2020.
\newblock Available at {\tt https://gmplib.org/}.

\bibitem[GG05]{Golomb-Gong}
Solomon~W. Golomb and Guang Gong.
\newblock {\em Signal design for good correlation}.
\newblock Cambridge University Press, Cambridge, 2005.

\bibitem[Gol67]{Golomb}
Solomon~W. Golomb.
\newblock {\em Shift register sequences}.
\newblock With portions co-authored by Lloyd R. Welch, Richard M. Goldstein,
  and Alfred W. Hales. Holden-Day, Inc., San Francisco,
  Calif.-Cambridge-Amsterdam, 1967.

\bibitem[Gol72]{Golay-72}
Marcel J.~E. Golay.
\newblock {A class of finite binary sequences with alternate autocorrelation
  values equal to zero.}
\newblock {\em {IEEE Trans. Inform. Theory}}, 18:449--450, 1972.

\bibitem[Gol75]{Golay-75}
Marcel J.~E. Golay.
\newblock {Hybrid low autocorrelation sequences.}
\newblock {\em {IEEE Trans. Inform. Theory}}, 21:460--462, 1975.

\bibitem[Gol82]{Golay-1982}
Marcel J.~E. Golay.
\newblock The merit factor of long low autocorrelation binary sequences.
\newblock {\em IEEE Trans. Inform. Theory}, 28(3):543--549, 1982.

\bibitem[Jed24]{Jedwab}
Jonathan Jedwab.
\newblock The mean and variance of the reciprocal merit factor of four classes
  of binary sequences.
\newblock arXiv:1911.11246, 2024.

\bibitem[JKS13]{Jedwab-Katz-Schmidt}
Jonathan Jedwab, Daniel~J. Katz, and Kai-Uwe Schmidt.
\newblock Littlewood polynomials with small {$L^4$} norm.
\newblock {\em Adv. Math.}, 241:127--136, 2013.

\bibitem[Kat16]{Katz16}
Daniel~J. Katz.
\newblock Aperiodic crosscorrelation of sequences derived from characters.
\newblock {\em IEEE Trans. Inform. Theory}, 62(9):5237--5259, 2016.

\bibitem[KR23]{Katz-Ramirez-limiting}
Daniel~J. Katz and Miriam~E. Ramirez.
\newblock Limiting moments of autocorrelation demerit factors of binary
  sequences.
\newblock arXiv:2307.14566, 2023.

\bibitem[Sar84]{Sarwate}
D.~V. Sarwate.
\newblock Mean-square correlation of shift-register sequences.
\newblock {\em Communications, Radar and Signal Processing, IEE Proceedings F},
  131(2):101--106, 1984.

\bibitem[Sch06]{Schroeder}
M.~R. Schroeder.
\newblock {\em Number theory in science and communication}, volume~7 of {\em
  Springer Series in Information Sciences}.
\newblock Springer-Verlag, Berlin, fourth edition, 2006.

\end{thebibliography}
\end{document}